\documentclass[runningheads]{llncs}

\usepackage[british]{babel}
\usepackage[T1]{fontenc}
\usepackage[utf8]{inputenc}
\usepackage{lmodern}
\usepackage[english]{isodate}

\usepackage[ruled,norelsize]{algorithm2e}

\usepackage{amsmath}
\usepackage{amssymb}
\usepackage{bm}

\DeclareMathOperator*{\argmin}{arg\,min}
\DeclareMathOperator{\sign}{sign}
\DeclareMathOperator{\integer}{int}

\DeclareMathOperator{\tekum}{\theta}

\DeclareMathOperator{\anchor}{anc}

\newcommand{\T}{\scalebox{.692}[1.0]{$\mathrm{T}$}}
\newcommand{\Tset}{\mathbb{T}}
\newcommand{\mdoubleplus}{\mathbin{+\mkern-10mu+}}

\newcommand{\cdotsx}{\mathinner{\cdotp\mkern-2mu\cdotp\mkern-2mu\cdotp}}
\newcommand{\ldotsx}{\mathinner{\ldotp\mkern-2mu\ldotp}}

\usepackage{graphicx}
\usepackage{pgf, tikz}
\usepackage{pgfplots}
\pgfplotsset{compat=1.17}
\usepgfplotslibrary{groupplots}
\usetikzlibrary{arrows.meta}
\usepackage{placeins}

\usepackage{subcaption}
\usepackage{pgfplots}

\usepackage{nameref}
\usepackage[hidelinks, unicode]{hyperref} 

\usepackage[style=numeric,sortcites=true]{biblatex}
\usepackage[autostyle=true]{csquotes}
\usepackage{enumitem}
\usepackage{listingsutf8}
\usepackage{orcidlink}
\usepackage{xcolor}
\usepackage{colortbl}
\usepackage{siunitx}
\usepackage{multicol}
\usepackage{multirow}
\usepackage{hhline}

\definecolor{sign}{HTML}{b02a2d}
\definecolor{direction}{HTML}{007b39}
\definecolor{regime}{HTML}{046f87} 
\definecolor{exponent}{HTML}{834f78} 
\definecolor{fraction}{HTML}{636363}
\definecolor{error}{HTML}{BD002A}
\definecolor{cellbg}{HTML}{EDEDED}

\definecolor{t-exponent}{HTML}{1f5dc2}

\definecolor{p-sign}{HTML}{FF5454}
\definecolor{p-regime}{HTML}{CC9966}
\definecolor{p-regime-term}{HTML}{996633}
\definecolor{p-exponent}{HTML}{0080FF}
\definecolor{p-fraction}{HTML}{000000}

\DeclareSIUnit{\parsec}{pc}
\DeclareSIUnit{\electronvolt}{eV}

\makeatletter
\def\lst@makecaption{%
  \def\@captype{table}%
  \@makecaption
}
\makeatother

\makeatletter
\newcommand*\notsotiny{%
  \@setfontsize\notsotiny{5}{5}%
}
\makeatother

\begin{document}

\title{Tekum: Balanced Ternary Tapered Precision Real Arithmetic}
\titlerunning{Tekum: Balanced Ternary Tapered Precision Arithmetic}
\author{Laslo Hunhold\,\orcidlink{0000-0001-8059-0298}}
\authorrunning{L. Hunhold}
\institute{%
	Parallel and Distributed Systems Group\\
	University of Cologne, Cologne, Germany\\
	\email{hunhold@uni-koeln.de}
}
\maketitle

\begin{abstract}
In light of recent hardware advances, it is striking that real arithmetic in balanced ternary logic has received almost no attention in the literature. This is particularly surprising given ternary logic's promising properties, which could open new avenues for energy-efficient computing and offer novel strategies for overcoming the memory wall.
\par
This paper revisits the concept of tapered precision arithmetic, as used in posit and takum formats, and introduces a new scheme for balanced ternary logic: tekum arithmetic. Several fundamental design challenges are addressed along the way. The proposed format is evaluated and shown to exhibit highly promising characteristics. In many respects, it outperforms both posits and takums. As ternary hardware matures, this work represents a crucial step toward unlocking the full potential of real-number computation in ternary systems, laying the groundwork for a new class of number formats designed from the ground up for a new category of next-generation hardware.
\end{abstract}

\section{Introduction}
Computer science is a comparatively young scientific discipline. The first 
programmable 
computers were constructed little more than eighty years ago, \textsc{Zuse}'s 
Z3 in 1941 and 
the ENIAC by \textsc{Mauchly} and \textsc{Eckert} in 1946. These early machines 
were based on 
the binary system, which has since shaped the foundations of computer science. 
The prevailing 
perception of computers is that, at their core, they operate through states of 
\enquote{on} 
or \enquote{off}. Decades of research and development have refined and 
optimised the design 
of binary computers.
\par
History, however, can narrow perspective. A lesser-known paradigm of 
computation is ternary 
logic. Instead of two states, ternary logic uses three; instead of binary 
digits (bits), one 
has ternary digits (trits), where a single trit contains $\log_2(3) \approx 
1.58$ bits of 
information. The motivation for considering such systems lies in the concept of 
\emph{radix 
economy}, the cost of representing a positive integer $n \in 
\mathbb{N}_1$ in a 
given base
$b \in \mathbb{R}_+$. This cost is defined as $C(b,n) := \lfloor \log_b(n) + 1 
\rfloor \cdot 
b$, namely the number of digits required to represent $n$ in base $b$, 
multiplied by the base 
$b$ itself \cite{2001-third_base}. Intuitively, larger bases produce shorter 
representations 
but incur greater hardware complexity, whereas smaller bases ease 
implementation but require 
more digits. It can be shown analytically that base $\mathrm{e} \approx 2.718$ 
minimises cost 
asymptotically, with $3$ being the nearest integer. If one assumes that base 
size correlates 
with hardware complexity, ternary arithmetic therefore offers the most balanced 
trade-off 
between circuit complexity and representational efficiency. By contrast, binary 
arithmetic 
simplifies circuit design, but at the expense of higher representational cost 
and greater 
communication overhead.
\par
Hardware investigations have compared binary and ternary implementations. While 
ternary 
circuits achieve comparable computational speed, they typically require more 
logic; for 
example, a ternary adder needs approximately \SI{62}{\percent} more circuitry, 
even 
accounting for the higher information density of trits 
\cite{2012-ternary_manifesto}. This 
finding has often been cited as a reason to dismiss ternary logic. Yet, such 
dismissal 
warrants reconsideration in light of modern computing trends: the dominant 
bottleneck today 
is not raw computational speed, but memory bandwidth 
\cite{1999-memory_wall,2024-memory_wall}. Under these conditions, ternary logic 
acquires 
renewed relevance. Recent advances, such as ternary deep neural networks and 
carbon nanotube 
transistors, which operative natively in ternary, further strengthen its value 
proposition.
\par
The full potential of ternary computation emerges when digits are chosen from 
$\{-1,0,1\}$ 
rather than $\{0,1,2\}$. This \emph{balanced ternary} system is unique to odd 
bases and may 
appear unfamiliar, given the dominance of even bases such as 2, 8, 10, 16, and 
64, or 
historically base 60. In balanced ternary, the digit $-1$ is commonly denoted 
by a special 
symbol, here written as $\T$. For instance, $2$ is represented as $1\T$ 
($1\cdot 3 + 
(-1)\cdot 1$), and $-4$ as $\T\T$ ($(-1)\cdot 3 + (-1)\cdot 1$). This 
representation has 
several appealing properties: integers are inherently signed, and perfectly 
symmetric around 
zero, without unlike two's complement binary integers. Negation is achieved 
simply by 
inverting digits, and rounding is the same as truncation, eliminating the need 
for carries. 
\textsc{Knuth} has described it as \enquote{perhaps the prettiest number system 
of all} 
\cite[207]{1997-knuth}.
\par
Balanced ternary has been studied mainly in the context of integers. 
Real-number 
representations, such as floating-point formats, remain surprisingly 
underexplored. One might 
expect that properties such as rounding by truncation could mitigate issues 
inherent in 
binary floating-point arithmetic, where carries complicate rounding. Yet the 
literature is 
sparse. The only notable proposal is \texttt{ternary27} by \textsc{O'Hare} 
\cite{2021-ternary27}, which is heavily inspired by IEEE 754. Although it 
achieves a 
favourable dynamic range, its design is inefficient, wasting many 
representations. This 
inefficiency is more detrimental in ternary than in binary arithmetic, as each 
trit carries 
more information than a bit. While the IEEE 854-1987 standard specifies 
radix-independent 
floating-point numbers, it does not address balanced bases.
\par
Given current hardware developments, it is undesirable for theoretical progress 
to lag behind 
practical advances. This paper therefore revisits balanced ternary real 
arithmetic and makes 
the 
following contributions. After introducing the fundamentals of balanced ternary 
in 
Section~\ref{sec:balanced_ternary}, we identify key design challenges in 
Section~\ref{sec:three_filters} that up to the design of takums 
\cite{2024-takum} were 
insurmountable. We then propose a new balanced ternary number format, 
\emph{tekum}, in Section~\ref{sec:tekum_definition}, and evaluate its 
properties in 
Section~\ref{sec:evaluation}, before drawing conclusions in 
Section~\ref{sec:conclusion}.
\section{Balanced Ternary}
\label{sec:balanced_ternary}
In this section we introduce balanced ternary and all the tools necessary
to define and analyse the subsequently introduced tekum arithmetic. We
follow a formal approach given the unintuitiveness compared to
standard binary logic. Additionally, this paper is more or less an
inaugural paper on balanced ternary real arithmetic, and it's helpful to
suggest a notation for this nascent field.
\begin{definition}[balanced ternary strings]
	Let $n \in \mathbb{N}_0$.
	The set of $n$-trit balanced ternary strings is defined as
	$\Tset_n := {\{ \T, 0, 1 \}}^n$ with $\T := -1$ .
	By convention $\Tset_n \ni \bm{t} := (\bm{t}_{n-1},\dots,\bm{t}_0) =:
	\bm{t}_{n-1}\cdots\bm{t}_0$.
\end{definition}
The first thing we notice is that, compared to posits and takums, we
denote ternary strings with bold lowercase letters instead of uppercase
letters. This is more in line with the common notation for vectors in
mathematics and computer science.
\begin{definition}[concatenation]
	Let $m,n \in \mathbb{N}_0$, $\bm{t} \in \Tset_m$ and $\bm{u} \in \Tset_n$.
	The concatenation of $\bm{t}$ and $\bm{u}$ is defined as
	$\bm{t} \mdoubleplus \bm{u} :=
	(\bm{t}_{m-1},\dots,\bm{t}_0,\bm{u}_{n-1},\dots,\bm{u}_0) \in \Tset_{m+n}$.
\end{definition}
This concatenation operator is distinct from the more relaxed, overloaded
concatenation notation with parentheses.
\begin{definition}[integer mapping]
	Let $n \in \mathbb{N}_0$. The integer mapping
	$\integer_n \colon \Tset_n \to \left\{ -\frac{1}{2} (3^n-1), \dots,
	\frac{1}{2} (3^n-1) \right\}$ is defined as
	$\integer_n(\bm{t}) := \sum_{i=0}^{n-1} \bm{t}_i 3^i$.
\end{definition}
Unlike with two's complement binary integers, there is no explicit sign bit.
Instead the sign is indicated by the most significant non-zero trit, and there
is no concept of an unsigned integer in balanced ternary.
The range of integer values follows from the finite geometric series obtained 
with the extremal values $\integer_n(\T\dots\T)$ and $\integer_n(1\dots1)$,
and is perfectly symmetric around zero (unlike two's complement integers).
\begin{definition}[negation]
	Let $n \in \mathbb{N}_0$ and $\bm{t} \in \Tset_n$. The negation of $\bm{t}$
	is defined as $-\bm{t} := (-\bm{t}_{n-1},\dots,-\bm{t}_0)$.
\end{definition}
By construction we can see that it holds $\integer_n(-\bm{t}) = 
-\integer_n(\bm{t})$, making it well-defined. The negation operation outlines 
another 
distinctive difference between two's complement binary and balanced ternary 
integers: Whereas to negate the underlying integral value you need to negate 
all bits and add one with the former, which is expensive in hardware due to the 
carry-chain, negation in balanced ternary is a very cheap, entrywise operation.
\begin{definition}[addition and subtraction]
	Let $n \in \mathbb{N}_0$ and $\bm{t}, \bm{u} \in \Tset_n$. The addition
	of $\bm{t}$ and $\bm{u}$ is defined with
	$s := \integer_n(\bm{t}) + \integer_n(\bm{u})$ as
	\begin{equation}
		\bm{t} + \bm{u} := \integer_n^\text{inv}\left(
		\begin{cases}
			s + \frac{1}{2} (3^n + 1) & s < -\frac{1}{2} (3^n - 1)\\
			s - \frac{1}{2} (3^n + 1) & s > \frac{1}{2} (3^n - 1)\\
			s & \text{otherwise}
		\end{cases}
		\right).
	\end{equation}
	The subtraction of $\bm{t}$ and $\bm{u}$ is defined as
	$\bm{t} - \bm{u} := \bm{t} + (-\bm{u})$.
\end{definition}
Despite the complex appearance, the addition is merely defined here as a fixed-size 
operation, taking two $n$-trit inputs and yielding an $n$-trit output, 
discarding any carries that might occur 
due to over- or underflow and using the inverse integer mapping to obtain the 
ternary string. If carries should be taken into account, the inputs 
need to be extended to $n+1$-trit strings before the addition. Especially 
notable is the definition of the subtraction: Given negation is so cheap, you 
can define it in terms of the addition.
\begin{definition}[modulus]
	Let $n \in \mathbb{N}_0$ and $\bm{t} \in \Tset_n$. The modulus of $\bm{t}$ is defined as
	\begin{equation}
		|\bm{t}| :=
		\begin{cases}
			-\bm{t} & \integer_n(\bm{t}) < 0\\
			\bm{t}  & \integer_n(\bm{t}) \ge 0.
		\end{cases}
	\end{equation}
\end{definition}
With these definitions in place we can proceed with the derivation of the
tekum arithmetic.
\section{The Three Filters}
\label{sec:three_filters}
In this section, we derive the central contribution of this paper: a new ternary real arithmetic.
As outlined in the introduction, the current state of the art in ternary real arithmetic is surprisingly
underdeveloped. This may be attributed in part to the relatively small number of arithmetic designers
worldwide, as well as the limited training and interest in ternary hardware to date. However, a closer
examination of recently established design criteria for number formats reveals three possible consecutive
filters that may have led previous researchers to attempt, and ultimately abandon, the pursuit of a
new format.
\par
The term filter is borrowed from Robert \textsc{Hanson}'s notion of the 
\enquote{Great Filter}, which refers
to the fundamental developmental barriers a civilisation must overcome, proposed to explain the
(apparent) absence of advanced extraterrestrial life in 
the universe \cite{1998-great_filter}. Here, we adapt 
the term
\enquote{filter} to describe conceptual or technical barriers that have likely hindered the development of a
proper ternary real arithmetic. We outline and reflect on these filters throughout the derivation, highlighting how each
is encountered and overcome in the formulation of our format.
\par
Fundamentally, the design of a real number format involves defining a mapping between discrete representations
(in this case, trit strings) and a corresponding set of real values. This set is constructed as a subset of the
real numbers~$\mathbb{R}$, often extended with additional non-number representations. The most common such extension is
the \emph{affinely extended real numbers}, $\mathbb{R} \cup \{-\infty, +\infty\}$. A less common alternative is the
\emph{projectively extended real numbers}, $\mathbb{R} \cup \{\infty\}$, where division by zero is well-defined.
\par
Further extensions typically include one or more 
special representations to
denote undefined or invalid values.
These are commonly referred to as $\mathrm{NaN}$ (not a 
number) or 
$\mathrm{NaR}$ (not a real), and are used to
represent non-real complex numbers or domain violations (e.g. $\ln(-1)$). While IEEE~754 floating-point numbers are
based on the affinely extended real numbers and have multiple $\mathrm{NaN}$ representations, formats such as posits
and takums instead exclude explicit representations of infinity and use a single $\mathrm{NaR}$ for both infinities
and non-real numbers.
\par
This was not always the case for posits. Early versions of the posit format only included a single infinity non-number
representation~\cite{posits-beating_floating-point-2017, lindstrom-2018}. However, this was later revised in the
draft standard~\cite{posits-standard-2022}, which removed the infinity representation in favour of a single
$\mathrm{NaR}$. The rationale for this decision is that the encoding scheme only permits a single special
non-numeric value, and the practical utility of being able to propagate $\mathrm{NaR}$ in computations to signal error
conditions was deemed to outweigh the mathematical elegance of a single infinity. The \emph{takum}
format subsequently adopted this revised approach~\cite{2024-takum}.
\par
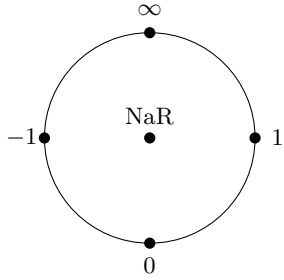
\begin{figure}[tbp]
	\centering
	\begin{tikzpicture}
		\tikzstyle{every node}=[font=\small]

		\def\circleradius{1.4cm}
		\def\beltdistance{0.3cm}
		\def\innerdistance{0.3cm}

		\draw (0,0) node[circle, inner sep=1.5pt, fill] {}
			circle [radius=\circleradius];
		\node[minimum height=0.4cm, inner sep=0] at (0, \beltdistance) {$\mathrm{NaR}$};
		\node[minimum width=0.5cm, minimum height=0.5cm, inner sep=0] at (0, -\beltdistance) {{\notsotiny}};

		\foreach \angle/\outerlabel/\innerlabel in {
			180/$-1$/,
			270/$0$/,
			0/$1$/,
	 		90/$\infty$/%
		} {
			\ifthenelse{\angle > 90 \AND \angle < 270}
				{\pgfmathsetmacro\rot{\angle + 180}}
				{
					\ifthenelse{\angle = 90 \OR \angle = 270}
						{\pgfmathsetmacro\rot{0}}
						{\pgfmathsetmacro\rot{\angle}}
				}

			\ifthenelse{\angle = 0 \OR \angle = 90 \OR \angle = 180 \OR \angle = 270}
				{\pgfmathsetmacro\noderadius{1.5pt}}
				{\pgfmathsetmacro\noderadius{1.0pt}}

			\draw (\angle: \circleradius)
				node[circle, inner sep=\noderadius, fill] {};

			\node[rotate=\rot, anchor=center, minimum width=0.4cm, minimum height=0.4cm, inner sep=0] at
				(\angle: {\circleradius + \beltdistance}) {\outerlabel};
			\node[rotate=\rot, anchor=center, minimum width=0.5cm, minimum height=0.5cm, inner sep=0] at
				(\angle: {\circleradius - \innerdistance}) {
					{\notsotiny \textcolor{fraction}{\innerlabel}}};
		}
	\end{tikzpicture}
	\caption{A visualisation of the real wheel algebra with the customary positioning of the bottom element
	$\mathrm{NaR}$ in the center as a \enquote{wheel hub} that distinguishes it from the projectively extended
	real numbers.}
	\label{fig:wheel}
\end{figure}%
The most mathematically complete and elegant alternative would be a set of represented values that includes
both $\infty$ and $\mathrm{NaR}$. Such a structure does exist in the form of the \emph{wheel algebra} applied to reals,
$\mathbb{R} \cup \{ \infty, \mathrm{NaR} \}$~\cite{wheels-2004}. In this framework, one defines $\infty := 1/0$
and $\mathrm{NaR} := 0/0$, where $\mathrm{NaR}$ is referred to as the \enquote{bottom element}. The wheel algebra supports
multiplication and inversion across all its elements, offering a closed and consistent model for arithmetic.
Figure~\ref{fig:wheel} illustrates a common visualisation of the wheel algebra: a circle representing the
extended real values, with $\mathrm{NaR}$ at the centre, resembling a wheel with a hub.
The numbers $1$ and $-1$ are explicitly included as fixed points under inversion (reflection of the circle across
the $x$-axis), just as $0$ and $\infty$ serve as fixed points under negation (reflection across the $y$-axis).
\par
In light of the mathematical elegance of the real wheel algebra, one might naturally ask why
formats such as posits and takums do not include both $\infty$ and $\mathrm{NaR}$. If we consider
discretising this set, we would ideally require an equal number of values in each quadrant of the wheel
so that both negation and inversion map represented values back into the set. In other words, after
assigning the five fundamental elements, $\mathrm{NaR}$, $-1$, $0$, $1$, and $\infty$, the remaining number
of unassigned states must be divisible by four.
\par
In binary representations, this condition can never be satisfied. The total number of states is a power of
two (and thus even), and subtracting the five fixed elements yields an odd number, which cannot be evenly
divided by four. Ternary strings, on the other hand, have a total number of states equal to a power of three,
which is always odd. Subtracting five from an odd number yields an even result, raising the question: could
this make the construction viable?
\par
\begin{figure}[tbp]
	\centering
	\begin{subfigure}[c]{0.32\textwidth}
		\centering
		\begin{tikzpicture}
			\def\circleradius{1.4cm}
			\def\beltdistance{0.3cm}
			\def\innerdistance{0.29cm}

			\draw (0,0) node[circle, inner sep=1.5pt, fill] {}
				circle [radius=\circleradius];
			\node[minimum height=0.4cm, inner sep=0] at (0, \beltdistance) {$\mathrm{NaR}$};
			\node[minimum width=0.5cm, minimum height=0.5cm, inner sep=0] at (0, -\beltdistance) {{\notsotiny \textcolor{fraction}{$\T$}}};

			\foreach \angle/\outerlabel/\innerlabel in {
				270/$0$/$0$,
		 		90/$\infty$/$1$%
			} {
				\ifthenelse{\angle > 90 \AND \angle < 270}
					{\pgfmathsetmacro\rot{\angle + 180}}
					{
						\ifthenelse{\angle = 90 \OR \angle = 270}
							{\pgfmathsetmacro\rot{0}}
							{\pgfmathsetmacro\rot{\angle}}
					}

				\ifthenelse{\angle = 0 \OR \angle = 90 \OR \angle = 180 \OR \angle = 270}
					{\pgfmathsetmacro\noderadius{1.5pt}}
					{\pgfmathsetmacro\noderadius{1.0pt}}

				\draw (\angle: \circleradius)
					node[circle, inner sep=\noderadius, fill] {};

				\node[rotate=\rot, anchor=center, minimum width=0.4cm, minimum height=0.4cm, inner sep=0] at
					(\angle: {\circleradius + \beltdistance}) {\outerlabel};
				\node[rotate=\rot, anchor=center, minimum width=0.5cm, minimum height=0.5cm, inner sep=0] at
					(\angle: {\circleradius - \innerdistance}) {
						{\notsotiny \textcolor{fraction}{\innerlabel}}};
			}
		\end{tikzpicture}
		\subcaption{$n=1$}
	\end{subfigure}
	\begin{subfigure}[c]{0.32\textwidth}
		\centering
		\begin{tikzpicture}
			\tikzstyle{every node}=[font=\small]

			\def\circleradius{1.4cm}
			\def\beltdistance{0.3cm}
			\def\innerdistance{0.3cm}

			\draw (0,0) node[circle, inner sep=1.5pt, fill] {}
				circle [radius=\circleradius];
			\node[minimum height=0.4cm, inner sep=0] at (0, \beltdistance) {$\mathrm{NaR}$};
			\node[minimum width=0.5cm, minimum height=0.5cm, inner sep=0] at (0, -\beltdistance) {{\notsotiny \textcolor{fraction}{$\T\T$}}};

			\foreach \angle/\outerlabel/\innerlabel in {
				135//$\T0$,
				180/$-1$/$\T1$,
				225//$0\T$,
				270/$0$/$00$,
				315//$01$,
				0/$1$/$1\T$,
		 		45//$10$,
		 		90/$\infty$/$11$%
			} {
				\ifthenelse{\angle > 90 \AND \angle < 270}
					{\pgfmathsetmacro\rot{\angle + 180}}
					{
						\ifthenelse{\angle = 90 \OR \angle = 270}
							{\pgfmathsetmacro\rot{0}}
							{\pgfmathsetmacro\rot{\angle}}
					}

				\ifthenelse{\angle = 0 \OR \angle = 90 \OR \angle = 180 \OR \angle = 270}
					{\pgfmathsetmacro\noderadius{1.5pt}}
					{\pgfmathsetmacro\noderadius{1.0pt}}

				\draw (\angle: \circleradius)
					node[circle, inner sep=\noderadius, fill] {};

				\node[rotate=\rot, anchor=center, minimum width=0.4cm, minimum height=0.4cm, inner sep=0] at
					(\angle: {\circleradius + \beltdistance}) {\outerlabel};
				\node[rotate=\rot, anchor=center, minimum width=0.5cm, minimum height=0.5cm, inner sep=0] at
					(\angle: {\circleradius - \innerdistance}) {
						{\notsotiny \textcolor{fraction}{\innerlabel}}};
			}
		\end{tikzpicture}
		\subcaption{$n=2$}
	\end{subfigure}
	\begin{subfigure}[c]{0.32\textwidth}
		\centering
		\begin{tikzpicture}
			\tikzstyle{every node}=[font=\small]

			\def\circleradius{1.4cm}
			\def\beltdistance{0.3cm}
			\def\innerdistance{0.31cm}

			\draw (0,0) node[circle, inner sep=1.5pt, fill] {}
				circle [radius=\circleradius];
			\node[minimum height=0.4cm, inner sep=0] at (0, \beltdistance) {$\mathrm{NaR}$};
			\node[minimum width=0.5cm, minimum height=0.5cm, inner sep=0] at (0, -\beltdistance) {{\notsotiny \textcolor{fraction}{$\T\T\T$}}};

			\foreach \angle/\outerlabel/\innerlabel in {
				105//$\T\T0$,
				120//$\T\T1$,
				135//$\T0\T$,
				150//$\T00$,
				165//$\T01$,
				180/$-1$/$\T1\T$,
				193//$\T10$,
				206//$\T11$,
				219//$0\T\T$,
				231//$0\T0$,
				244//$0\T1$,
				257//$00\T$,
				270/$0$/$000$,
				283//$001$,
				296//$01\T$,
				309//$010$,
				321//$011$,
				334//$1\T\T$,
				347//$1\T0$,
				0/$1$/$1\T1$,
		 		15//$10\T$,
		 		30//$100$,
		 		45//$101$,
		 		60//$11\T$,
		 		75//$110$,
		 		90/$\infty$/$111$%
			} {
				\ifthenelse{\angle > 90 \AND \angle < 270}
					{\pgfmathsetmacro\rot{\angle + 180}}
					{
						\ifthenelse{\angle = 90 \OR \angle = 270}
							{\pgfmathsetmacro\rot{0}}
							{\pgfmathsetmacro\rot{\angle}}
					}
				\ifthenelse{\angle = 0 \OR \angle = 90 \OR \angle = 180 \OR \angle = 270}
					{\pgfmathsetmacro\noderadius{1.5pt}}
					{\pgfmathsetmacro\noderadius{1.0pt}}
		
				\draw (\angle: \circleradius)
					node[circle, inner sep=\noderadius, fill] {};
		
				\node[rotate=\rot, anchor=center, minimum width=0.4cm, minimum height=0.4cm, inner sep=0] at
					(\angle: {\circleradius + \beltdistance}) {\outerlabel};
				\node[rotate=\rot, anchor=center, minimum width=0.5cm, minimum height=0.5cm, inner sep=0] at
					(\angle: {\circleradius - \innerdistance}) {
						{\notsotiny \textcolor{fraction}{\innerlabel}}};
			}
		\end{tikzpicture}
		\subcaption{$n=3$}
	\end{subfigure}
	\caption{
		Mapping of ternary strings of lengths $n \in \{1,2,3\}$ to the real wheel algebra.
	}
	\label{fig:wheel-discrete}
\end{figure}%
We assign the bottom element, $\mathrm{NaR}$, to the lexicographically smallest balanced ternary integer,
$\T\cdots\T$, ensuring that it is also the smallest element under comparison. This approach is consistent with
the total ordering predicate introduced in the 2019 revision of the IEEE~754 standard, which similarly defines
$\mathrm{NaN}$ as smaller than all real numbers \cite[§5.10]{ieee754-2019}. Proceeding in increasing order, we
assign the remaining values such that $0\cdots0$ always maps to zero, and the lexicographically largest balanced
ternary integer, $1\cdots1$, maps to $\infty$. The result of this assignment scheme for $n \in \{1,2,3\}$ is
illustrated in Figure~\ref{fig:wheel-discrete}.
\subsection{Filter 1: Asymmetry}
The attentive reader may have already observed that while the assignment illustrated in Figure~\ref{fig:wheel-discrete}
works well for $n \in \{1,2\}$, it introduces an asymmetry for $n = 3$: the upper quadrants contain five elements,
whereas the lower quadrants contain six. This irregularity marks the first filter in our derivation process.
It stems from the fact that ternary strings do not yield state counts as conveniently structured as in the binary case,
where the total number of states is $2^n$ for $n$ bits, and the number of assignable values per quadrant,
namely $(2^n - 4)/4$ (after reserving the four special elements $\mathrm{NaR}$, $-1$, $0$, and $1$), is always an
integer, specifically $2^{n-2} - 1$.
\par
In contrast, no such regularity emerges for $3^n - 5$ (or $3^n - 4$ if we omit $\infty$). Nevertheless,
a single viable path is shown in the following
\begin{proposition}
	Let $n \in \mathbb{N}_0$. It holds $4 \mid (3^{2n} - 5)$ and $4 \nmid (3^{n} - 4)$.
\end{proposition}
\begin{proof}
	With $a \equiv b \pmod c \rightarrow a^n \equiv b^n \pmod c$ for
	$a,b,c \in \mathbb{Z}$ via 
	\cite[Lemma~5.2c]{1976-apostol-number_theory} 
	it holds
	\begin{align}
		4 \mid (3^{2n} - 5)
			&\Leftrightarrow 3^{2n} - 5 \equiv 0 \pmod 4\\
		&\Leftrightarrow 9^{n} \equiv 5 \pmod 4\\
		&\Leftrightarrow 9^{n} \equiv 1 \pmod 4\\
		&\Leftrightarrow 9^{n} \equiv 1^n \pmod 4\\
		&\Leftarrow 9 \equiv 1 \pmod 4.
	\end{align}
	We can also show
	\begin{align}
		4 \mid (3^{n} - 4) \Leftrightarrow 3^n - 4 \equiv 0 \pmod 4
		\Leftrightarrow 3^n \equiv 0 \pmod 4 \Leftrightarrow 4 \mid 3^n,
	\end{align}
	however $3^n$ is never divisible by four as it's odd.\qed
\end{proof}
This proposition demonstrates two key points. First, it confirms that a meaningful ternary discretisation of the real wheel
algebra is indeed possible, provided that $n$ is even. Second, it highlights that reserving only four special
elements, the approach taken by posits and takums, can never yield a symmetric distribution, and thus is
fundamentally unsuitable for this purpose. We therefore proceed under the assumption that $n$ is even and
use the assignment as outlined earlier.
\subsection{Filter 2: Misfit Tool}
Despite having overcome the first filter, and having settled on the general approach to
restrict $n$ to even numbers, we have yet to define the mapping rule for the values within
each quadrant beyond the fundamental elements. This mapping should ideally follow the tapered
precision paradigm. However, here we encounter the second filter: the prefix strings used
in posit arithmetic cannot be translated directly into ternary logic, rendering them an
ill-suited tool for ternary arithmetic. Similarly, takum's fixed-size regime offers no
straightforward adaptation to the ternary case.
\par
Let us first consider a positive trit string $\bm{t} \in \{ 0\cdots 01, \dots, 1\cdots 10 \}$. Our aim is to
process this string in such a way that it yields a meaningful mapping to the positive real numbers. 
The initial observation is that, given the extreme cases $0\cdots 0$ (representing $0$) and $1\cdots 1$
(representing $\infty$), and knowing that all quadrants contain the same number of elements, the value $1$
must lie exactly in the middle. The corresponding ternary sequence for $1$ is $1\T\cdots1\T$, since
$1\T + 1\T = 11$ holds, a property that extends to any trit string of even length.
Given that $1$ corresponds to exponent zero, a natural approach is to subtract $1\T\cdots1\T$ from $\bm{t}$.
This yields zero when $\bm{t} = 1\T\cdots1\T$, a positive value when $\bm{t}$ lies in the upper quadrant $(1,\infty)$,
and a negative value when $\bm{t}$ is in the lower quadrant $(0,1)$.
\par
By construction, the expression $(\bm{t} - 1\T\cdots 1\T)$ lies within the integer set
$\{ \T1\cdots \T10\T, \dots, 1\T\cdots 1\T01 \}$. We can now employ the takum approach, using a fixed number
of trits to encode the so-called \enquote{regime} value, which serves as the parameter for a variable-length
exponent. Although this may seem counter-intuitive given, after all, the regime values appear to be evenly spaced,
the key insight is that each additional trit used for the exponent reduces the number of trits available for
encoding the fraction. As a result, the number density decreases by a factor of three for each trit allocated
to the exponent, thus realising tapered precision.
\par
The number of trits allocated to the regime represents a trade-off: fewer regime trits allow for higher maximum
precision, but at the expense of the number of representable regime values. We now explore this design choice in
more detail. If we assign the first two trits of $\bm{t} - 1\T\cdots 1\T$ to represent the regime, the regime trits
fall within the range $\{ \T 1, \dots, 1\T \}$, corresponding to the integer interval $\{-2, \dots, 2\}$. In
contrast, assigning the first four trits as regime results in a range of $\{ \T1\T1, \dots, 1\T1\T \}$, mapping to
the integer interval $\{-20, \dots, 20\}$, a range far exceeding the requirements for encoding exponent lengths.
A compromise is to use three regime trits, yielding a regime range of $\{ \T1\T, \dots, 1\T1 \}$, corresponding
to the integer interval $\{-7, \dots, 7\}$. This choice offers a balanced trade-off between exponent flexibility
and precision. The precise role of the regime value in variable-length exponent encoding will be discussed later.
\par
The final aspect to address is the treatment of negative trit strings
$\bm{t} \in \{ \T\cdots\T0, \dots, 0\cdots0\T \}$. A straightforward approach is to compute the modulus
$|\bm{t}|$ and process the result. This is justified by the symmetry of the representation: both upper quadrants
$(-\infty, -1)$ and $(1, \infty)$ correspond to the same positive exponents, while the lower quadrants $(-1, 0)$ and
$(0, 1)$ share the same negative exponents. The only difference lies in the sign. Thus, for all trit strings
$\bm{t} \in \Tset_n$, we shall compute $|\bm{t}| - 1\T\cdots 1\T$ and assign the first three trits to be the regime
trits $\textcolor{regime}{\bm{r}}$. As we \enquote{anchor} the trits onto a workable representation,
we give the following corresponding
\begin{definition}[anchor function]
	Let $n \in 2\mathbb{N}_0$. The anchor function $\anchor_n \colon \Tset_n \to \Tset_n$ is defined as
	$\anchor_n(\bm{t}) = |\bm{t}| - 1\T\cdots1\T$.
\end{definition}
\begin{figure}[tb]
	\centering
	\begin{tikzpicture}
		\tikzstyle{every node}=[font=\small]

		\def\circleradius{4.0cm}
		\def\beltdistance{0.35cm}
		\def\innerdistance{1.1cm}
		\def\textraiseoffset{-1.5mm}

		\draw (0,0) node[circle, inner 
		sep=1.5pt, fill] {}
			circle [radius=\circleradius];
		\node[minimum height=0.4cm, inner 
		sep=0] at (0, \beltdistance) 
		{$\mathrm{NaR}$};
		\node[minimum width=0.5cm, minimum 
		height=0.5cm, inner sep=0] 
		at (0, -\beltdistance) {{\notsotiny 
		\textcolor{fraction}{$\T\T\T\T$}}};

		\foreach \angle/\outerlabel/\innerlabel 
		in {
			 94.5//\raisebox{\textraiseoffset}{$\T\T\T0
			 \rightarrow 
			 \textcolor{regime}{1\T0}1 
			 \rightarrow 
			 6\phantom{+}$},
			 99.0//\raisebox{\textraiseoffset}{$\T\T\T1
			 \rightarrow 
			 \textcolor{regime}{1\T0}0 
			 \rightarrow 
			 6\phantom{+}$},
			103.5//\raisebox{\textraiseoffset}{$\T\T0\T
			\rightarrow 
			\textcolor{regime}{1\T0}\T 
			\rightarrow 
			6\phantom{+}$},
			108.0//\raisebox{\textraiseoffset}{$\T\T00
			\rightarrow 
			\textcolor{regime}{1\T\T}1 
			\rightarrow 
			5\phantom{+}$},
			112.5//\raisebox{\textraiseoffset}{$\T\T01
			\rightarrow 
			\textcolor{regime}{1\T\T}0 
			\rightarrow 
			5\phantom{+}$},
			117.0//\raisebox{\textraiseoffset}{$\T\T1\T
			\rightarrow 
			\textcolor{regime}{1\T\T}\T 
			\rightarrow 
			5\phantom{+}$},
			121.5//\raisebox{\textraiseoffset}{$\T\T10
			\rightarrow 
			\textcolor{regime}{011}1 
			\rightarrow 
			4\phantom{+}$},
			126.0//\raisebox{\textraiseoffset}{$\T\T11
			\rightarrow 
			\textcolor{regime}{011}0 
			\rightarrow 
			4\phantom{+}$},
			130.5//\raisebox{\textraiseoffset}{$\T0\T\T
			\rightarrow 
			\textcolor{regime}{011}\T 
			\rightarrow 
			4\phantom{+}$},
			135.0//\raisebox{\textraiseoffset}{$\T0\T0
			\rightarrow 
			\textcolor{regime}{010}1 
			\rightarrow 
			3\phantom{+}$},
			139.5//\raisebox{\textraiseoffset}{$\T0\T1
			\rightarrow 
			\textcolor{regime}{010}0 
			\rightarrow 
			3\phantom{+}$},
			144.0//\raisebox{\textraiseoffset}{$\T00\T
			\rightarrow 
			\textcolor{regime}{010}\T 
			\rightarrow 
			3\phantom{+}$},
			148.5//\raisebox{\textraiseoffset}{$\T000
			\rightarrow 
			\textcolor{regime}{01\T}1 
			\rightarrow 
			2\phantom{+}$},
			153.0//\raisebox{\textraiseoffset}{$\T001
			\rightarrow 
			\textcolor{regime}{01\T}0 
			\rightarrow 
			2\phantom{+}$},
			157.5//\raisebox{\textraiseoffset}{$\T01\T
			\rightarrow 
			\textcolor{regime}{01\T}\T 
			\rightarrow 
			2\phantom{+}$},
			162.0//\raisebox{\textraiseoffset}{$\T010
			\rightarrow 
			\textcolor{regime}{001}1 
			\rightarrow 
			1\phantom{+}$},
			166.5//\raisebox{\textraiseoffset}{$\T011
			\rightarrow 
			\textcolor{regime}{001}0 
			\rightarrow 
			1\phantom{+}$},
			171.0//\raisebox{\textraiseoffset}{$\T1\T\T
			\rightarrow 
			\textcolor{regime}{001}\T 
			\rightarrow 
			1\phantom{+}$},
			175.5//\raisebox{\textraiseoffset}{$\T1\T0
			\rightarrow 
			\textcolor{regime}{000}1 
			\rightarrow 
			0\phantom{+}$},
			180.0/$-1$/\raisebox{\textraiseoffset}{$\T1\T1
			\rightarrow 
			\textcolor{regime}{000}0 
			\rightarrow 
			0\phantom{+}$},
			184.5//\raisebox{\textraiseoffset}{$\T10\T
			\rightarrow 
			\textcolor{regime}{000}\T 
			\rightarrow 
			0\phantom{+}$},
			189.0//\raisebox{\textraiseoffset}{$\T100
			\rightarrow 
			\textcolor{regime}{00\T}1 
			\rightarrow 1-$},
			193.5//\raisebox{\textraiseoffset}{$\T101
			\rightarrow 
			\textcolor{regime}{00\T}0 
			\rightarrow 1-$},
			198.0//\raisebox{\textraiseoffset}{$\T11\T
			\rightarrow 
			\textcolor{regime}{00\T}\T 
			\rightarrow 1-$},
			202.5//\raisebox{\textraiseoffset}{$\T110
			\rightarrow 
			\textcolor{regime}{0\T1}1 
			\rightarrow 2-$},
			207.0//\raisebox{\textraiseoffset}{$\T111
			\rightarrow 
			\textcolor{regime}{0\T1}0 
			\rightarrow 2-$},
			211.5//\raisebox{\textraiseoffset}{$0\T\T\T
			\rightarrow 
			\textcolor{regime}{0\T1}\T 
			\rightarrow 2-$},
			216.0//\raisebox{\textraiseoffset}{$0\T\T0
			\rightarrow 
			\textcolor{regime}{0\T0}1 
			\rightarrow 3-$},
			220.5//\raisebox{\textraiseoffset}{$0\T\T1
			\rightarrow 
			\textcolor{regime}{0\T0}0 
			\rightarrow 3-$},
			225.0//\raisebox{\textraiseoffset}{$0\T0\T
			\rightarrow 
			\textcolor{regime}{0\T0}\T 
			\rightarrow 3-$},
			229.5//\raisebox{\textraiseoffset}{$0\T00
			\rightarrow 
			\textcolor{regime}{0\T\T}1 
			\rightarrow 4-$},
			234.0//\raisebox{\textraiseoffset}{$0\T01
			\rightarrow 
			\textcolor{regime}{0\T\T}0 
			\rightarrow 4-$},
			238.5//\raisebox{\textraiseoffset}{$0\T1\T
			\rightarrow 
			\textcolor{regime}{0\T\T}\T 
			\rightarrow 
			4-$},
			243.0//\raisebox{\textraiseoffset}{$0\T10
			\rightarrow 
			\textcolor{regime}{\T11}1 
			\rightarrow 5-$},
			247.5//\raisebox{\textraiseoffset}{$0\T11
			\rightarrow 
			\textcolor{regime}{\T11}0 
			\rightarrow 5-$},
			252.0//\raisebox{\textraiseoffset}{$00\T\T
			\rightarrow 
			\textcolor{regime}{\T11}\T 
			\rightarrow 5-$},
			256.5//\raisebox{\textraiseoffset}{$00\T0
			\rightarrow 
			\textcolor{regime}{\T10}1 
			\rightarrow 6-$},
			261.0//\raisebox{\textraiseoffset}{$00\T1
			\rightarrow 
			\textcolor{regime}{\T10}0 
			\rightarrow 6-$},
			265.5//\raisebox{\textraiseoffset}{$000\T
			\rightarrow 
			\textcolor{regime}{\T10}\T 
			\rightarrow 6-$},
			270.0/$0$/\raisebox{\textraiseoffset}{$\phantom{\phantom{+}0
			 \leftarrow 0000 \leftarrow \ } 
			 0000$},
			274.5//\raisebox{\textraiseoffset}{$-6
			 \leftarrow 
			\textcolor{regime}{\T10}\T  
			\leftarrow 0001$},
			279.0//\raisebox{\textraiseoffset}{$-6
			 \leftarrow 
			\textcolor{regime}{\T10}0   
			\leftarrow 001\T$},
			283.5//\raisebox{\textraiseoffset}{$-6
			 \leftarrow 
			\textcolor{regime}{\T10}1   
			\leftarrow 0010$},
			288.0//\raisebox{\textraiseoffset}{$-5
			 \leftarrow 
			\textcolor{regime}{\T11}\T  
			\leftarrow 0011$},
			292.5//\raisebox{\textraiseoffset}{$-5
			 \leftarrow 
			\textcolor{regime}{\T11}0   
			\leftarrow 01\T\T$},
			297.0//\raisebox{\textraiseoffset}{$-5
			 \leftarrow 
			\textcolor{regime}{\T11}1   
			\leftarrow 01\T0$},
			301.5//\raisebox{\textraiseoffset}{$-4
			 \leftarrow 
			\textcolor{regime}{0\T\T}\T 
			\leftarrow 01\T1$},
			306.0//\raisebox{\textraiseoffset}{$-4
			 \leftarrow 
			\textcolor{regime}{0\T\T}0  
			\leftarrow 010\T$},
			310.5//\raisebox{\textraiseoffset}{$-4
			 \leftarrow 
			\textcolor{regime}{0\T\T}1  
			\leftarrow 0100$},
			315.0//\raisebox{\textraiseoffset}{$-3
			 \leftarrow 
			\textcolor{regime}{0\T0}\T  
			\leftarrow 0101$},
			319.5//\raisebox{\textraiseoffset}{$-3
			 \leftarrow 
			\textcolor{regime}{0\T0}0   
			\leftarrow 011\T$},
			324.0//\raisebox{\textraiseoffset}{$-3
			 \leftarrow 
			\textcolor{regime}{0\T0}1   
			\leftarrow 0110$},
			328.5//\raisebox{\textraiseoffset}{$-2
			 \leftarrow 
			\textcolor{regime}{0\T1}\T  
			\leftarrow 0111$},
			333.0//\raisebox{\textraiseoffset}{$-2
			 \leftarrow 
			\textcolor{regime}{0\T1}0   
			\leftarrow 1\T\T\T$},
			337.5//\raisebox{\textraiseoffset}{$-2
			 \leftarrow 
			\textcolor{regime}{0\T1}1   
			\leftarrow 1\T\T0$},
			342.0//\raisebox{\textraiseoffset}{$-1
			 \leftarrow 
			\textcolor{regime}{00\T}\T  
			\leftarrow 1\T\T1$},
			346.5//\raisebox{\textraiseoffset}{$-1
			 \leftarrow 
			\textcolor{regime}{00\T}0   
			\leftarrow 1\T0\T$},
			351.0//\raisebox{\textraiseoffset}{$-1
			 \leftarrow 
			\textcolor{regime}{00\T}1   
			\leftarrow 1\T00$},
			355.5//\raisebox{\textraiseoffset}{$\phantom{+}0
			\leftarrow 
			\textcolor{regime}{000}\T   
			\leftarrow 
			1\T01$},
			  0.0/$1$/\raisebox{\textraiseoffset}{$\phantom{+}0
			  \leftarrow 
			  \textcolor{regime}{000}0 
			  \leftarrow 
			  1\T1\T$},
			  4.5//\raisebox{\textraiseoffset}{$\phantom{+}0
			  \leftarrow 
			  \textcolor{regime}{000}1    
			  \leftarrow 
			  1\T10$},
			  9.0//\raisebox{\textraiseoffset}{$\phantom{+}1
			  \leftarrow 
			  \textcolor{regime}{001}\T   
			  \leftarrow 
			  1\T11$},
			 13.5//\raisebox{\textraiseoffset}{$\phantom{+}1
			 \leftarrow 
			 \textcolor{regime}{001}0    
			 \leftarrow 
			 10\T\T$},
			 18.0//\raisebox{\textraiseoffset}{$\phantom{+}1
			 \leftarrow 
			 \textcolor{regime}{001}1    
			 \leftarrow 
			 10\T0$},
			 22.5//\raisebox{\textraiseoffset}{$\phantom{+}2
			 \leftarrow 
			 \textcolor{regime}{01\T}\T  
			 \leftarrow 
			 10\T1$},
			 27.0//\raisebox{\textraiseoffset}{$\phantom{+}2
			 \leftarrow 
			 \textcolor{regime}{01\T}0   
			 \leftarrow 
			 100\T$},
			 31.5//\raisebox{\textraiseoffset}{$\phantom{+}2
			 \leftarrow 
			 \textcolor{regime}{01\T}1   
			 \leftarrow 
			 1000$},
			 36.0//\raisebox{\textraiseoffset}{$\phantom{+}3
			 \leftarrow 
			 \textcolor{regime}{010}\T   
			 \leftarrow 
			 1001$},
			 40.5//\raisebox{\textraiseoffset}{$\phantom{+}3
			 \leftarrow 
			 \textcolor{regime}{010}0    
			 \leftarrow 
			 101\T$},
			 45.0//\raisebox{\textraiseoffset}{$\phantom{+}3
			 \leftarrow 
			 \textcolor{regime}{010}1    
			 \leftarrow 
			 1010$},
			 49.5//\raisebox{\textraiseoffset}{$\phantom{+}4
			 \leftarrow 
			 \textcolor{regime}{011}\T   
			 \leftarrow 
			 1011$},
			 54.0//\raisebox{\textraiseoffset}{$\phantom{+}4
			 \leftarrow 
			 \textcolor{regime}{011}0    
			 \leftarrow 
			 11\T\T$},
			 58.5//\raisebox{\textraiseoffset}{$\phantom{+}4
			 \leftarrow 
			 \textcolor{regime}{011}1    
			 \leftarrow 
			 11\T0$},
			 63.0//\raisebox{\textraiseoffset}{$\phantom{+}5
			 \leftarrow 
			 \textcolor{regime}{1\T\T}\T 
			 \leftarrow 
			 11\T1$},
			 67.5//\raisebox{\textraiseoffset}{$\phantom{+}5
			 \leftarrow 
			 \textcolor{regime}{1\T\T}0  
			 \leftarrow 
			 110\T$},
			 72.0//\raisebox{\textraiseoffset}{$\phantom{+}5
			 \leftarrow 
			 \textcolor{regime}{1\T\T}1  
			 \leftarrow 
			 1100$},
			 76.5//\raisebox{\textraiseoffset}{$\phantom{+}6
			 \leftarrow 
			 \textcolor{regime}{1\T0}\T  
			 \leftarrow 
			 1101$},
			 81.0//\raisebox{\textraiseoffset}{$\phantom{+}6
			 \leftarrow 
			 \textcolor{regime}{1\T0}0   
			 \leftarrow 
			 111\T$},
			 85.5//\raisebox{\textraiseoffset}{$\phantom{+}6
			 \leftarrow 
			 \textcolor{regime}{1\T0}1   
			 \leftarrow 
			 1110$},
	 		 90.0/$\infty$/\raisebox{\textraiseoffset}{$\phantom{\phantom{+}0
	 		  \leftarrow 0000 \leftarrow \ 
	 		  } 1111$}%
		} {
			\ifthenelse{\lengthtest{\angle 
			pt > 90pt} \AND 
			\lengthtest{\angle pt < 270pt}}
				{\pgfmathsetmacro\rot{\angle
				 + 
				180}\pgfmathsetmacro\textrot{\angle
				 + 180}}
				{
					\ifthenelse{\lengthtest{\angle
					 pt = 
					90pt} \OR 
					\lengthtest{\angle
					 pt = 
					270pt}}
						{\pgfmathsetmacro\rot{\angle}\pgfmathsetmacro\textrot{0}}
						{\pgfmathsetmacro\rot{\angle}\pgfmathsetmacro\textrot{\angle}}
				}
			\ifthenelse{\lengthtest{\angle 
			pt = 0pt} \OR 
			\lengthtest{\angle pt = 90pt} 
			\OR
					\lengthtest{\angle
					 pt = 180pt} 
					\OR 
					\lengthtest{\angle
					 pt = 270pt}}
				{\pgfmathsetmacro\noderadius{1.5pt}}
				{\pgfmathsetmacro\noderadius{1.0pt}}
	
			\draw (\angle: \circleradius)
				node[circle, inner 
				sep=\noderadius, fill] 
				{};
	
			\node[rotate=\textrot, 
			anchor=center, minimum 
			width=0.4cm, minimum 
			height=0.4cm, inner sep=0] at
				(\angle: {\circleradius 
				+ \beltdistance}) 
				{\outerlabel};
			\node[rotate=\rot, 
			anchor=center, minimum 
			width=0.5cm, 
			minimum height=0.2cm, inner 
			sep=0] at
				(\angle: {\circleradius 
				- \innerdistance}) {
					{\notsotiny 
					\textcolor{fraction}{\innerlabel}}};
		}
	\end{tikzpicture}
	\caption{
		Mapping of ternary strings $\bm{t}$ of 
		length $4$ to the real 
		wheel algebra.
		The values of $\anchor_4(\bm{t})$ are 
		given, with the first
		three trits, designated as the regime 
		trits 
		$\textcolor{regime}{\bm{r}}$,
		highlighted accordingly. The 
		corresponding regime values $r$ 
		are also indicated, partially
		with suffixed signs for better visual 
		consistency.
	}
	\label{fig:wheel-discrete-4}
\end{figure}
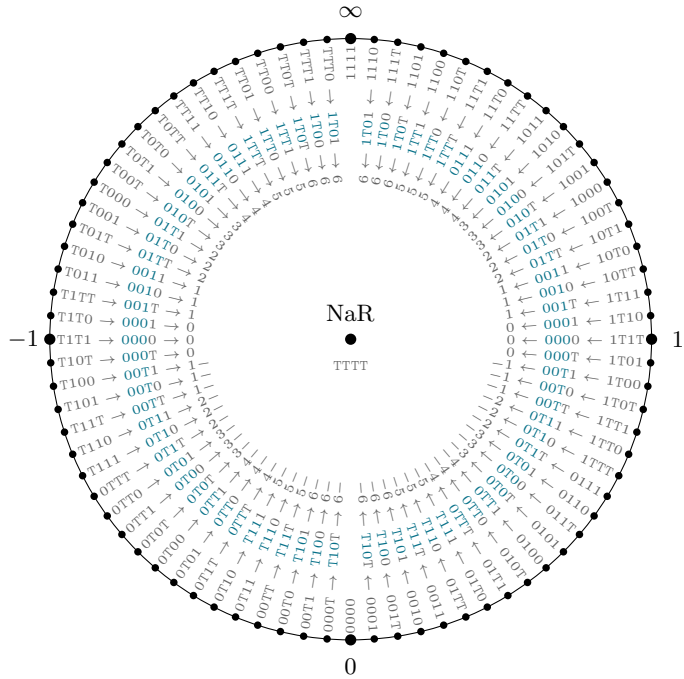%
An illustration of the proposed mapping for $n=4$ is provided in
Figure~\ref{fig:wheel-discrete-4}. At this point, we can confidently conclude that we have found a way to
make the tool, namely the takum approach, fit again, overcoming the second filter in the process.
\subsection{Filter 3: Excess}
\begin{table}[p]
	\caption{Overview of mappings from regime values to exponent trit counts.}
	\label{tab:regime_mappings}
	\scriptsize
	\begin{center}
		\bgroup
		\def\arraystretch{1.15}
		\setlength{\tabcolsep}{0.1em}
		\begin{tabular}{|l||l|l|l|l|l|l|l|l|l|}
		\cline{2-10}
		\multicolumn{1}{c|}{}                                         & $r$                                & $0$                                   & $1$                                  & $2$                                   & $3$                                  & $4$                                  & $5$                                    & $6$                                   & $7$                                      \\ \hhline{~|=|=|=|=|=|=|=|=|=|} 
		\multicolumn{1}{c|}{}                                         & \multirow{2}{*}{$\anchor(\bm{t})$} & $\textcolor{regime}{000}\T\cdotsx\T$  & $\textcolor{regime}{001}\T\cdotsx\T$ & $\textcolor{regime}{01\T}\T\cdotsx\T$ & $\textcolor{regime}{010}\T\cdotsx\T$ & $\textcolor{regime}{011}\T\cdotsx\T$ & $\textcolor{regime}{1\T\T}\T\cdotsx\T$ & $\textcolor{regime}{1\T0}\T\cdotsx\T$ & $\textcolor{regime}{1\T1}\T\cdotsx\T$    \\ \hhline{-|~--------|} 
		c(r)                                                          &                                    & $\textcolor{regime}{000}1\cdotsx1$    & $\textcolor{regime}{001}1\cdotsx1$   & $\textcolor{regime}{01\T}1\cdotsx1$   & $\textcolor{regime}{010}1\cdotsx1$   & $\textcolor{regime}{011}1\cdotsx1$   & $\textcolor{regime}{1\T\T}1\cdotsx1$   & $\textcolor{regime}{1\T0}1\cdotsx1$   & $\textcolor{regime}{1\T1}\T\cdotsx1\T01$ \\ \hhline{|=||=|=|=|=|=|=|=|=|=|} 
		\multirow{7}{*}{$|r|$}                                        & $c$                                & $0$                                   & $1$                                  & $2$                                   & $3$                                  & $4$                                  & $5$                                    & $6$                                   & $7$                                      \\ \cline{2-10} 
		                                                              & $\bm{e}_\text{min}$                &                                       & $\T$                                 & $\T\T$                                & $\T\T\T$                             & $\T\T\T\T$                           & $\T\T\T\T\T$                           & $\T\T\T\T\T\T$                        & $\T\T\T\T\T\T\T$                         \\ \cline{2-10} 
		                                                              & $\bm{e}_\text{max}$                &                                       & $1$                                  & $11$                                  & $111$                                & $1111$                               & $11111$                                & $111111$                              & $\T1\T1\T1\T$                            \\ \cline{2-10} 
		                                                              & $\integer(\bm{e})$                 & $0\ldotsx0$                           & $-1\ldotsx1$                         & $-4\ldotsx4$                          & $-13\ldotsx13$                       & $-40\ldotsx40$                       & $-121\ldotsx121$                       & $-364\ldotsx364$                      & $-1093\ldotsx{-547}$                     \\ \cline{2-10} 
		                                                              & $b$                                & $0$                                   & $2$                                  & $8$                                   & $26$                                 & $80$                                 & $242$                                  & $728$                                 & $2186$                                   \\ \cline{2-10} 
		                                                              & $e$                                & $0\ldotsx0$                           & $1\ldotsx3$                          & $4\ldotsx12$                          & $13\ldotsx39$                        & $40\ldotsx120$                       & $121\ldotsx363$                        & $364\ldotsx1092$                      & $1093\ldotsx1639$                        \\ \cline{2-10} 
		                                                              & $\lg(3^{e})$                       & $0\ldotsx0$                           & $0.5\ldotsx1.4$                      & $1.9\ldotsx5.7$                       & $6.2\ldotsx19$                       & $19\ldotsx57$                        & $58\ldotsx173$                         & $174\ldotsx521$                       & $521\ldotsx782$                          \\ \hhline{|=||=|=|=|=|=|=|=|=|=|} 
		\multirow{7}{*}{$\begin{matrix}\max(0,\\|r|-1)\end{matrix}$}  & $c$                                & $0$                                   & $0$                                  & $1$                                   & $2$                                  & $3$                                  & $4$                                    & $5$                                   & $6$                                      \\ \cline{2-10} 
		                                                              & $\bm{e}_\text{min}$                &                                       &                                      & $\T$                                  & $\T\T$                               & $\T\T\T$                             & $\T\T\T\T$                             & $\T\T\T\T\T$                          & $\T\T\T\T\T\T$                           \\ \cline{2-10} 
		                                                              & $\bm{e}_\text{max}$                &                                       &                                      & $1$                                   & $11$                                 & $111$                                & $1111$                                 & $11111$                               & $\T1\T1\T1$                              \\ \cline{2-10} 
		                                                              & $\integer(\bm{e})$                 & $0\ldotsx0$                           & $0\ldotsx0$                          & $-1\ldotsx1$                          & $-4\ldotsx4$                         & $-13\ldotsx13$                       & $-40\ldotsx40$                         & $-121\ldotsx121$                      & $-364\ldotsx{-182}$                      \\ \cline{2-10} 
		                                                              & $b$                                & $0$                                   & $1$                                  & $3$                                   & $9$                                  & $27$                                 & $81$                                   & $243$                                 & $729$                                    \\ \cline{2-10} 
		                                                              & $e$                                & $0\ldotsx0$                           & $1\ldotsx1$                          & $2\ldotsx4$                           & $5\ldotsx13$                         & $14\ldotsx40$                        & $41\ldotsx121$                         & $122\ldotsx364$                       & $365\ldotsx547$                          \\ \cline{2-10} 
		                                                              & $\lg(3^{e})$                       & $0\ldotsx0$                           & $0.5\ldotsx0.5$                      & $1.0\ldotsx1.9$                       & $2.4\ldotsx6.2$                      & $6.7\ldotsx19$                       & $20\ldotsx58$                          & $58\ldotsx174$                        & $174\ldotsx261$                          \\ \hhline{|=||=|=|=|=|=|=|=|=|=|} 
		\multirow{7}{*}{alpha}                                        & $c$                                & $0$                                   & $0$                                  & $1$                                   & $1$                                  & $2$                                  & $3$                                    & $4$                                   & $5$                                      \\ \cline{2-10} 
		                                                              & $\bm{e}_\text{min}$                &                                       &                                      & $\T$                                  & $\T$                                 & $\T\T$                               & $\T\T\T$                               & $\T\T\T\T$                            & $\T\T\T\T\T$                             \\ \cline{2-10} 
		                                                              & $\bm{e}_\text{max}$                &                                       &                                      & $1$                                   & $1$                                  & $11$                                 & $111$                                  & $1111$                                & $\T1\T1\T$                               \\ \cline{2-10} 
		                                                              & $\integer(\bm{e})$                 & $0\ldotsx0$                           & $0\ldotsx0$                          & $-1\ldotsx1$                          & $-1\ldotsx1$                         & $-4\ldotsx4$                         & $-13\ldotsx13$                         & $-40\ldotsx40$                        & $-121\ldotsx{-61}$                       \\ \cline{2-10} 
		                                                              & $b$                                & $0$                                   & $1$                                  & $3$                                   & $6$                                  & $12$                                 & $30$                                   & $84$                                  & $246$                                    \\ \cline{2-10} 
		                                                              & $e$                                & $0\ldotsx0$                           & $1\ldotsx1$                          & $2\ldotsx4$                           & $5\ldotsx7$                          & $8\ldotsx16$                         & $17\ldotsx43$                          & $44\ldotsx124$                        & $125\ldotsx185$                          \\ \cline{2-10} 
		                                                              & $\lg(3^{e})$                       & $0\ldotsx0$                           & $0.5\ldotsx0.5$                      & $1.0\ldotsx1.9$                       & $2.4\ldotsx3.3$                      & $3.8\ldotsx7.6$                      & $8.1\ldotsx21$                         & $21\ldotsx59$                         & $60\ldotsx88$                            \\ \hhline{|=||=|=|=|=|=|=|=|=|=|} 
		\multirow{7}{*}{$\begin{matrix}\max(0,\\|r|-2)\end{matrix}$}  & $c$                                & $0$                                   & $0$                                  & $0$                                   & $1$                                  & $2$                                  & $3$                                    & $4$                                   & $5$                                      \\ \cline{2-10} 
		                                                              & $\bm{e}_\text{min}$                &                                       &                                      &                                       & $\T$                                 & $\T\T$                               & $\T\T\T$                               & $\T\T\T\T$                            & $\T\T\T\T\T$                             \\ \cline{2-10} 
		                                                              & $\bm{e}_\text{max}$                &                                       &                                      &                                       & $1$                                  & $11$                                 & $111$                                  & $1111$                                & $\T1\T1\T$                               \\ \cline{2-10} 
		                                                              & $\integer(\bm{e})$                 & $0\ldotsx0$                           & $0\ldotsx0$                          & $0\ldotsx0$                           & $-1\ldotsx1$                         & $-4\ldotsx4$                         & $-13\ldotsx13$                         & $-40\ldotsx40$                        & $-121\ldotsx{-61}$                       \\ \cline{2-10} 
		                                                              & $b$                                & $0$                                   & $1$                                  & $2$                                   & $4$                                  & $10$                                 & $28$                                   & $82$                                  & $244$                                    \\ \cline{2-10} 
		                                                              & $e$                                & $0\ldotsx0$                           & $1\ldotsx1$                          & $2\ldotsx2$                           & $3\ldotsx5$                          & $6\ldotsx14$                         & $15\ldotsx41$                          & $42\ldotsx122$                        & $123\ldotsx183$                          \\ \cline{2-10} 
		                                                              & $\lg(3^{e})$                       & $0\ldotsx0$                           & $0.5\ldotsx0.5$                      & $1.0\ldotsx1.0$                       & $1.4\ldotsx2.4$                      & $2.9\ldotsx6.7$                      & $7.1\ldotsx20$                         & $20\ldotsx58$                         & $59\ldotsx87$                            \\ \hhline{|=||=|=|=|=|=|=|=|=|=|} 
		\multirow{7}{*}{beta}                                         & $c$                                & $0$                                   & $0$                                  & $0$                                   & $1$                                  & $1$                                  & $2$                                    & $3$                                   & $4$                                      \\ \cline{2-10} 
		                                                              & $\bm{e}_\text{min}$                &                                       &                                      &                                       & $\T$                                 & $\T$                                 & $\T\T$                                 & $\T\T\T$                              & $\T\T\T\T$                               \\ \cline{2-10} 
		                                                              & $\bm{e}_\text{max}$                &                                       &                                      &                                       & $1$                                  & $1$                                  & $11$                                   & $111$                                 & $\T1\T1$                                 \\ \cline{2-10} 
		                                                              & $\integer(\bm{e})$                 & $0\ldotsx0$                           & $0\ldotsx0$                          & $0\ldotsx0$                           & $-1\ldotsx1$                         & $-1\ldotsx1$                         & $-4\ldotsx4$                           & $-13\ldotsx13$                        & $-40\ldotsx{-20}$                        \\ \cline{2-10} 
		                                                              & $b$                                & $0$                                   & $1$                                  & $2$                                   & $4$                                  & $7$                                  & $13$                                   & $31$                                  & $85$                                     \\ \cline{2-10} 
		                                                              & $e$                                & $0\ldotsx0$                           & $1\ldotsx1$                          & $2\ldotsx2$                           & $3\ldotsx5$                          & $6\ldotsx8$                          & $9\ldotsx17$                           & $18\ldotsx44$                         & $45\ldotsx65$                            \\ \cline{2-10} 
		                                                              & $\lg(3^{e})$                       & $0\ldotsx0$                           & $0.5\ldotsx0.5$                      & $1.0\ldotsx1.0$                       & $1.4\ldotsx2.4$                      & $2.9\ldotsx3.8$                      & $4.3\ldotsx8.1$                        & $8.6\ldotsx21$                        & $21\ldotsx31$                            \\ \hhline{|=||=|=|=|=|=|=|=|=|=|} 
		\multirow{7}{*}{$\begin{matrix}\max(0,\\|r|-3)\end{matrix}$}  & $c$                                & $0$                                   & $0$                                  & $0$                                   & $0$                                  & $1$                                  & $2$                                    & $3$                                   & $4$                                      \\ \cline{2-10} 
		                                                              & $\bm{e}_\text{min}$                &                                       &                                      &                                       &                                      & $\T$                                 & $\T\T$                                 & $\T\T\T$                              & $\T\T\T\T$                               \\ \cline{2-10} 
		                                                              & $\bm{e}_\text{max}$                &                                       &                                      &                                       &                                      & $1$                                  & $11$                                   & $111$                                 & $\T1\T1$                                 \\ \cline{2-10} 
		                                                              & $\integer(\bm{e})$                 & $0\ldotsx0$                           & $0\ldotsx0$                          & $0\ldotsx0$                           & $0\ldotsx0$                          & $-1\ldotsx1$                         & $-4\ldotsx4$                           & $-13\ldotsx13$                        & $-40\ldotsx{-20}$                        \\ \cline{2-10} 
		                                                              & $b$                                & $0$                                   & $1$                                  & $2$                                   & $3$                                  & $5$                                  & $11$                                   & $29$                                  & $83$                                     \\ \cline{2-10} 
		                                                              & $e$                                & $0\ldotsx0$                           & $1\ldotsx1$                          & $2\ldotsx2$                           & $3\ldotsx3$                          & $4\ldotsx6$                          & $7\ldotsx15$                           & $16\ldotsx42$                         & $43\ldotsx63$                            \\ \cline{2-10} 
		                                                              & $\lg(3^{e})$                       & $0\ldotsx0$                           & $0.5\ldotsx0.5$                      & $1.0\ldotsx1.0$                       & $1.4\ldotsx1.4$                      & $1.9\ldotsx2.9$                      & $3.3\ldotsx7.2$                        & $7.6\ldotsx20$                        & $21\ldotsx30$                            \\ \hhline{|=||=|=|=|=|=|=|=|=|=|} 
		\multirow{7}{*}{gamma}                                        & $c$                                & $0$                                   & $0$                                  & $0$                                   & $0$                                  & $1$                                  & $1$                                    & $2$                                   & $3$                                      \\ \cline{2-10} 
		                                                              & $\bm{e}_\text{min}$                &                                       &                                      &                                       &                                      & $\T$                                 & $\T$                                   & $\T\T$                                & $\T\T\T$                                 \\ \cline{2-10} 
		                                                              & $\bm{e}_\text{max}$                &                                       &                                      &                                       &                                      & $1$                                  & $1$                                    & $11$                                  & $\T1\T$                                  \\ \cline{2-10} 
		                                                              & $\integer(\bm{e})$                 & $0\ldotsx0$                           & $0\ldotsx0$                          & $0\ldotsx0$                           & $0\ldotsx0$                          & $-1\ldotsx1$                         & $-1\ldotsx1$                           & $-4\ldotsx4$                          & $-13\ldotsx{-7}$                         \\ \cline{2-10} 
		                                                              & $b$                                & $0$                                   & $1$                                  & $2$                                   & $3$                                  & $5$                                  & $8$                                    & $14$                                  & $32$                                     \\ \cline{2-10} 
		                                                              & $e$                                & $0\ldotsx0$                           & $1\ldotsx1$                          & $2\ldotsx2$                           & $3\ldotsx3$                          & $4\ldotsx6$                          & $7\ldotsx9$                            & $10\ldotsx18$                         & $19\ldotsx25$                            \\ \cline{2-10} 
		                                                              & $\lg(3^{e})$                       & $0\ldotsx0$                           & $0.5\ldotsx0.5$                      & $1.0\ldotsx1.0$                       & $1.4\ldotsx1.4$                      & $1.9\ldotsx2.9$                      & $3.3\ldotsx4.3$                        & $4.8\ldotsx8.6$                       & $9.1\ldotsx12$                           \\ \cline{1-10} 

		\end{tabular}
		\egroup
	\end{center}
\end{table}
Having set the number of regime trits to three, after excluding the alternative choices
of two and four, the next stage is to actually investigate the format we obtain
from this approach. As mentioned before, a straightforward approach is to take the
modulus of the regime value $r \in \{-7,\dots,7\}$ and use it as the count $c$ of
exponent trits $\bm{e}$ following after the regime bits. With a careful selection of biases
$b$ added to the value $\integer(\bm{e})$ represented by the exponent trits we obtain a
consecutive sequence of exponent values $e$. Please refer to the first row block
of Table~\ref{tab:regime_mappings} to observe the derivation of biases and outcome
in terms of exponent values for each regime.
\par
We notice a big issue: This approach yields an obviously excessive dynamic range of
$10^{\pm 782}$. But what is a good general purpose dynamic range? The first work setting a
rough bound is by \textsc{Quevedo}, who remarks \enquote{No he señalado límite al valor del
exponente; pero es evidente que en todos los cálculos usuales será menor de ciento}
(I have not set a limit on the value of the [base-10] exponent, but it is clear that in all usual
calculations it will be less than one hundred; translation by the author) \cite[582\psq]{quevedo-1914}.
A more thorough investigation in \cite[Section~1.2]{2024-takum} makes the case
for a dynamic range of $10^{\pm55}$. While one can go into the multitude of
application-specific arithmetic, for which there is almost no lower bound on
dynamic range, the author would like to make the case that if someone builds
general-purpose, dedicated hardware, the arithmetic shall be as well.
\par
To tame the excessive dynamic range yielded by choosing $c(r) = |r|$, we consider the
parametric liberties we have in this case. The mapping $c(r)$ shall be non-decreasing and
start at zero. An additional justifiable constraint is that consecutive counts
should only differ by zero or one. These constraints limit the choices of $c(r)$,
and the first six are given in Table~\ref{tab:regime_mappings}: Three mappings
are straightforward shifts, where we only \enquote{start counting} after a certain
threshold is exceeded. Until then the regimes only 
represent a single exponent 
value.
Given are also three more special mappings called \enquote{alpha}, \enquote{beta} and
\enquote{gamma}, where after the skip the counting is not linear, but actually 
tapered itself by repeating the count one for two times.
\par
Overall we can see, in terms of dynamic ranges, a wide selection, going all the way down
to $10^{\pm 12}$. Given we explored all possible, reasonable choices of $c(r)$ within
this range we can say with confidence that this parametric space has been
exhausted. While $\max(0, |r| - 1)$ still has an excessive dynamic range of
$10^{\pm 261}$, the choices beta, $\max(0, |r| - 3)$ and gamma have insufficient
general purpose dynamic ranges $10^{\pm 31}$, $10^{\pm 30}$ and $10^{\pm 12}$ respectively.
Only remaining as candidates are alpha and $\max(0, |r| - 2)$ with similar 
dynamic
ranges $10^{\pm88}$ and $10^{\pm87}$ respectively, both satisfying both the 
\textsc{Quevedo}
limit of $10^{100}$ and the lower bound $10^{55}$ derived in \cite{2024-takum}.
The final choice falls on $\max(0, |r| - 2)$, as it has higher precision for
smaller values, which should always take precedence over the rarer \enquote{outer} numbers.
\par
\section{Tekum Definition}\label{sec:tekum_definition}
In this section, we introduce the tekum format, building upon the observations 
made in the previous section. The term \enquote{tekum} is derived from a 
combination of \enquote{ternary} and \enquote{takum}, reflecting the two 
foundational pillars of its design. First, the format is based on a balanced 
ternary representation. Second, it follows the design principles of takum, most 
notably its limited dynamic range. The name \enquote{takum} itself originates 
from the Icelandic phrase \enquote{takmarkað umfang}, which translates to 
\enquote{limited range}. On this basis, we obtain the following format:
\begin{definition}[tekum encoding]\label{def:tekum}
Let $n \in 2\mathbb{N}_1$ with $n \ge 8$. Any
$\bm{t} \in \Tset_n$ with 
$\textcolor{regime}{\bm{r}} \mdoubleplus
\textcolor{exponent}{\bm{e}} \mdoubleplus \textcolor{fraction}{\bm{f}} := 
\anchor_n(\bm{t})$ of the form
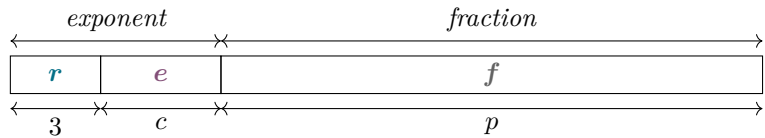
\begin{center}
	\begin{tikzpicture}
		\draw[<->] (0.0, 0.7) -- (2.8, 0.7) node[above,pos=.5] {exponent};
		\draw[<->] (2.8, 0.7) -- (10.0, 0.7) node[above,pos=.5] {fraction};

		\draw (0.0,0  ) rectangle (1.2,0.5) node[pos=.5] 
		{$\textcolor{regime}{\bm{r}}$};
		\draw (1.2,0  ) rectangle (2.8,0.5) node[pos=.5] 
		{$\textcolor{exponent}{\bm{e}}$};
		\draw (2.8,0  ) rectangle (10.0,0.5) node[pos=.5] 
		{$\textcolor{fraction}{\bm{f}}$};

		\draw[<->] (0.0, -0.2) -- (1.2, -0.2) node[below,pos=.5] {$3$};
		\draw[<->] (1.2, -0.2) -- (2.8, -0.2) node[below,pos=.5] {$c$};
		\draw[<->] (2.8, -0.2) -- (10.0, -0.2) node[below,pos=.5] {$p$};
	\end{tikzpicture}
\end{center}
with \emph{regime trits} $\textcolor{regime}{\bm{r}}$, \emph{exponent trits}
$\textcolor{exponent}{\bm{e}}$, \emph{fraction trits} 
$\textcolor{fraction}{\bm{f}}$, and
\begin{align}
	s &:= \sign(\integer_n(\bm{t})) &\colon \parbox{2.9cm}{sign}\\
	r &:= \integer_3(\textcolor{regime}{\bm{r}})
		\in \{-7,\dots,7\}
		&\colon \parbox{2.9cm}{regime value}\\
	c &:= \max(0, |r| - 2) \in \{0,\dots,5\} &\colon \parbox{2.9cm}{exponent trit count}\\
	p &:= n - c - 3 \in \{ n-8,\dots,n-3 \} &\colon
		\parbox{2.9cm}{fraction trit count}\\
	b &:= \sign(r) \cdot \left\lfloor 3^{|r|-2} + 1 \right\rfloor\notag\\ &\phantom{:}=
		\sign(r) \cdot {(0,1,2,4,10,28,82,244)}_{|r|} &\colon \parbox{2.9cm}{exponent bias}\\
	e &:= \integer_c(\textcolor{exponent}{\bm{e}})+ b \in \{-183,\dots,183\}
		&\colon \parbox{2.9cm}{exponent value}\\
	f &:= 3^{-p} \integer_p(\textcolor{fraction}{\bm{f}}) \in (-0.5,0.5)
		&\colon \parbox{2.9cm}{fraction value}
\end{align}
encodes the tekum value
\begin{equation}\label{eq:tekum}
	\tekum_n(\bm{t})
	:= \begin{cases}
		\mathrm{NaR}
			& \textcolor{regime}{\bm{r}} \mdoubleplus
			\textcolor{exponent}{\bm{e}} \mdoubleplus 
			\textcolor{fraction}{\bm{f}} = \T\cdots\T\\
		0
			& \textcolor{regime}{\bm{r}} \mdoubleplus
			\textcolor{exponent}{\bm{e}} \mdoubleplus 
			\textcolor{fraction}{\bm{f}} = 0\cdots 0\\
		\infty
			& \textcolor{regime}{\bm{r}} \mdoubleplus
			\textcolor{exponent}{\bm{e}} \mdoubleplus 
			\textcolor{fraction}{\bm{f}} = 1\cdots 1\\
		s \cdot (1+f) \cdot 3^{e} & \text{otherwise}
	\end{cases}
\end{equation}
with $\tekum \colon \Tset_n \mapsto \{ \mathrm{NaR}, 0, \infty \} \cup
\pm\left(0.5 \cdot 3^{-183},1.5\cdot 3^{183}\right)$.
Without loss of generality, any trit string shorter than 8 trits with an
even, positive number of trits is also included in the definition by matching the special
cases ($\mathrm{NaR}$, $0$, and $\infty$) respectively and expanding 
$\textcolor{regime}{\bm{r}} \mdoubleplus
\textcolor{exponent}{\bm{e}} \mdoubleplus 
\textcolor{fraction}{\bm{f}}$ with zeros in non-special cases. The case $n=1$, 
which only covers the special cases, is also
trivially included.
By convention, $\mathrm{NaR}$ and $\infty$ are 
defined to be 
smaller and larger than any other represented 
value, respectively.
\end{definition}
We also introduce a tekum colour scheme, prioritising uniformity in 
both lightness and chroma within the perceptually uniform OKLCH colour
space \cite{2023-colour}. Detailed colour definitions are delineated in 
Table~\ref{tab:colour_scheme}.
\begin{table}[tbp]
	\caption{Overview of the tekum arithmetic colour scheme.}
	\label{tab:colour_scheme}
	\small
	\begin{center}
		\bgroup
		\def\arraystretch{1.2}
		\setlength{\tabcolsep}{0.18em}
		\begin{tabular}{| l || l || l | l | l |}
			\hline
			colour & identifier & OKLCH & CIELab & HEX (sRGB)\\
			\hline\hline
			\cellcolor{regime} & regime & $(50\%, 0.09, 220)$ &
			$(42.44, -20.42, -21.18)$ & \texttt{\#046F87}\\
			\hline
			\cellcolor{exponent} & exponent & $(50\%, 0.09, 335)$ &
			$(40.72, 27.23, -13.9)$ & \texttt{\#834F78}\\
			\hline
			\cellcolor{fraction} & fraction & $(50\%, 0.00, 0)$ &
			$(42.00, 0.00, 0.00)$ & \texttt{\#636363}\\
			\hline
		\end{tabular}
		\egroup
	\end{center}
\end{table}
\par
Just as with the posit and takum formats, the tekum format can also be 
expressed as a 
logarithmic number system by renaming the fraction trits as \emph{mantissa 
trits}, denoted by 
\textcolor{fraction}{$\bm{m}$}, and correspondingly the fraction value as $m$, 
the 
\emph{mantissa value}. In this case, the fourth expression in \eqref{eq:tekum} 
becomes $s \cdot 3^{e+m}$. Since the transition from binary to ternary is 
already a sufficiently radical step, we refrain from introducing a different 
base for the logarithmic format and retain base~3. In this way, both the linear 
and logarithmic tekum formats share the same overall numerical properties, with 
the logarithmic form serving merely as an implementation detail. Within the 
scope of this work we do not pursue logarithmic tekums further, but simply note 
the possibility.
\par
Particular attention may be drawn to the floating-point representation 
$(1+f)\cdot 3^e$. 
Although it may appear counterintuitive, it is indeed correct. This becomes 
evident when 
considering its value range, $(0.5\cdot 3^e,\, 1.5\cdot 3^e)$. Notably, the 
upper bound $1.5 
\cdot 3^e$ is equal to $0.5 \cdot 3^{e+1}$, thereby connecting seamlessly to 
the adjacent 
range $(0.5\cdot 3^{e+1},\, 1.5\cdot 3^{e+1})$.
\par
\begin{table}[p]
	\caption{Tekum decoding table for $n=4$, pruned to positive numbers.}
	\label{tab:example}
	\small
	\begin{center}
		\bgroup
		\def\arraystretch{1.04}
		\setlength{\tabcolsep}{0.18em}
		\begin{tabular}{|l|r|r|l|l|r|r|l|r|r|r|l|r|r|r|}
		\hline
		$\bm{t}$             & $\integer_4(\bm{t})$ & $s$ & 
		$\anchor_4(\bm{t})$ & 
		$\textcolor{regime}{\bm{r}}$   & $r$  & $o$ & 
		$\textcolor{exponent}{\bm{e}}$   & $b$   & $e$   & $p$ & 
		$\textcolor{fraction}{\bm{f}}$ & $f$    & $1+f$  & 
		$\tekum_4(\bm{t})$ \\ 
		\hline\hline
		$000\T$   & $-1$     & $-1$ & 
		$\textcolor{regime}{\T10}\textcolor{exponent}{\T}$           & 
		$\T10$  & $-6$ 
		& $4$ & $\T000$ & $-82$ & $-109$ & $0$ &            & 
		$0.0$             & 
		$1.0$            & \num{-9.9e-53} \\ \hline
		$0000$    & $0$      &      
		&                                                              
		&         
		&      &     &         &       &        &     &            
		&                   &                  & $0$            \\ 
		\hline
		$0001$    & $1$      & $1$  & 
		$\textcolor{regime}{\T10}\textcolor{exponent}{\T}$           & 
		$\T10$  & $-6$ 
		& $4$ & $\T000$ & $-82$ & $-109$ & $0$ &            & 
		$0.0$             & 
		$1.0$            & \num{9.9e-53}  \\ \hline
		$001\T$   & $2$      & $1$  & 
		$\textcolor{regime}{\T10}\textcolor{exponent}{0}$            & 
		$\T10$  & $-6$ 
		& $4$ & $0000$  & $-82$ & $-82$  & $0$ &            & 
		$0.0$             & 
		$1.0$            & \num{7.5e-40}  \\ \hline
		$0010$    & $3$      & $1$  & 
		$\textcolor{regime}{\T10}\textcolor{exponent}{1}$            & 
		$\T10$  & $-6$ 
		& $4$ & $1000$  & $-82$ & $-55$  & $0$ &            & 
		$0.0$             & 
		$1.0$            & \num{5.7e-27}  \\ \hline
		$0011$    & $4$      & $1$  & 
		$\textcolor{regime}{\T11}\textcolor{exponent}{\T}$           & 
		$\T11$  & $-5$ 
		& $3$ & $\T00$  & $-28$ & $-37$  & $0$ &            & 
		$0.0$             & 
		$1.0$            & \num{2.2e-18}  \\ \hline
		$01\T\T$  & $5$      & $1$  & 
		$\textcolor{regime}{\T11}\textcolor{exponent}{0}$            & 
		$\T11$  & $-5$ 
		& $3$ & $000$   & $-28$ & $-28$  & $0$ &            & 
		$0.0$             & 
		$1.0$            & \num{4.4e-14}  \\ \hline
		$01\T0$   & $6$      & $1$  & 
		$\textcolor{regime}{\T11}\textcolor{exponent}{1}$            & 
		$\T11$  & $-5$ 
		& $3$ & $100$   & $-28$ & $-19$  & $0$ &            & 
		$0.0$             & 
		$1.0$            & \num{8.6e-10}  \\ \hline
		$01\T1$   & $7$      & $1$  & 
		$\textcolor{regime}{0\T\T}\textcolor{exponent}{\T}$          & 
		$0\T\T$ & $-4$ 
		& $2$ & $\T0$   & $-10$ & $-13$  & $0$ &            & 
		$0.0$             & 
		$1.0$            & \num{6.3e-7}   \\ \hline
		$010\T$   & $8$      & $1$  & 
		$\textcolor{regime}{0\T\T}\textcolor{exponent}{0}$           & 
		$0\T\T$ & $-4$ 
		& $2$ & $00$    & $-10$ & $-10$  & $0$ &            & 
		$0.0$             & 
		$1.0$            & \num{1.7e-5}   \\ \hline
		$0100$    & $9$      & $1$  & 
		$\textcolor{regime}{0\T\T}\textcolor{exponent}{1}$           & 
		$0\T\T$ & $-4$ 
		& $2$ & $10$    & $-10$ & $-7$   & $0$ &            & 
		$0.0$             & 
		$1.0$            & \num{4.6e-4}   \\ \hline
		$0101$    & $10$     & $1$  & 
		$\textcolor{regime}{0\T0}\textcolor{exponent}{\T}$           & 
		$0\T0$  & $-3$ 
		& $1$ & $\T$    & $-4$  & $-5$   & $0$ &            & 
		$0.0$             & 
		$1.0$            & \num{4.1e-3}   \\ \hline
		$011\T$   & $11$     & $1$  & 
		$\textcolor{regime}{0\T0}\textcolor{exponent}{0}$            & 
		$0\T0$  & $-3$ 
		& $1$ & $0$     & $-4$  & $-4$   & $0$ &            & 
		$0.0$             & 
		$1.0$            & \num{1.2e-2}   \\ \hline
		$0110$    & $12$     & $1$  & 
		$\textcolor{regime}{0\T0}\textcolor{exponent}{1}$            & 
		$0\T0$  & $-3$ 
		& $1$ & $1$     & $-4$  & $-3$   & $0$ &            & 
		$0.0$             & 
		$1.0$            & \num{3.7e-2}   \\ \hline
		$0111$    & $13$     & $1$  & 
		$\textcolor{regime}{0\T1}\textcolor{fraction}{\T}$           & 
		$0\T1$  & $-2$ 
		& $0$ &         & $-2$  & $-2$   & $1$ & $\T$       & 
		$-0.\overline{3}$ & 
		$0.\overline{7}$ & \num{7.4e-2}   \\ \hline
		$1\T\T\T$ & $14$     & $1$  & 
		$\textcolor{regime}{0\T1}\textcolor{fraction}{0}$            & 
		$0\T1$  & $-2$ 
		& $0$ &         & $-2$  & $-2$   & $1$ & $0$        & 
		$0.0$             & 
		$1.0$            & \num{1.1e-1}   \\ \hline
		$1\T\T0$  & $15$     & $1$  & 
		$\textcolor{regime}{0\T1}\textcolor{fraction}{1}$            & 
		$0\T1$  & $-2$ 
		& $0$ &         & $-2$  & $-2$   & $1$ & $1$        & 
		$0.\overline{3}$  & 
		$1.\overline{3}$ & \num{1.5e-1}   \\ \hline
		$1\T\T1$  & $16$     & $1$  & 
		$\textcolor{regime}{00\T}\textcolor{fraction}{\T}$           & 
		$00\T$  & $-1$ 
		& $0$ &         & $-1$  & $-1$   & $1$ & $\T$       & 
		$-0.\overline{3}$ & 
		$0.\overline{7}$ & \num{2.2e-1}   \\ \hline
		$1\T0\T$  & $17$     & $1$  & 
		$\textcolor{regime}{00\T}\textcolor{fraction}{0}$            & 
		$00\T$  & $-1$ 
		& $0$ &         & $-1$  & $-1$   & $1$ & $0$        & 
		$0.0$             & 
		$1.0$            & \num{3.3e-1}   \\ \hline
		$1\T00$   & $18$     & $1$  & 
		$\textcolor{regime}{00\T}\textcolor{fraction}{1}$            & 
		$00\T$  & $-1$ 
		& $0$ &         & $-1$  & $-1$   & $1$ & $1$        & 
		$0.\overline{3}$  & 
		$1.\overline{3}$ & \num{4.4e-1}   \\ \hline
		$1\T01$   & $19$     & $1$  & 
		$\textcolor{regime}{000}\textcolor{fraction}{\T}$            & 
		$000$   & $0$  
		& $0$ &         & $0$   & $0$    & $1$ & $\T$       & 
		$-0.\overline{3}$ & 
		$0.\overline{7}$ & \num{6.7e-1}   \\ \hline
		$1\T1\T$  & $20$     & $1$  & 
		$\textcolor{regime}{000}\textcolor{fraction}{0}$             & 
		$000$   & $0$  
		& $0$ &         & $0$   & $0$    & $1$ & $0$        & 
		$0.0$             & 
		$1.0$            & \num{1.0e0}    \\ \hline
		$1\T10$   & $21$     & $1$  & 
		$\textcolor{regime}{000}\textcolor{fraction}{1}$             & 
		$000$   & $0$  
		& $0$ &         & $0$   & $0$    & $1$ & $1$        & 
		$0.\overline{3}$  & 
		$1.\overline{3}$ & \num{1.3e0}    \\ \hline
		$1\T11$   & $22$     & $1$  & 
		$\textcolor{regime}{001}\textcolor{fraction}{\T}$            & 
		$001$   & $1$  
		& $0$ &         & $1$   & $1$    & $1$ & $\T$       & 
		$-0.\overline{3}$ & 
		$0.\overline{7}$ & \num{2.0e0}    \\ \hline
		$10\T\T$  & $23$     & $1$  & 
		$\textcolor{regime}{001}\textcolor{fraction}{0}$             & 
		$001$   & $1$  
		& $0$ &         & $1$   & $1$    & $1$ & $0$        & 
		$0.0$             & 
		$1.0$            & \num{3.0e0}    \\ \hline
		$10\T0$   & $24$     & $1$  & 
		$\textcolor{regime}{001}\textcolor{fraction}{1}$             & 
		$001$   & $1$  
		& $0$ &         & $1$   & $1$    & $1$ & $1$        & 
		$0.\overline{3}$  & 
		$1.\overline{3}$ & \num{4.0e0}    \\ \hline
		$10\T1$   & $25$     & $1$  & 
		$\textcolor{regime}{01\T}\textcolor{fraction}{\T}$           & 
		$01\T$  & $2$  
		& $0$ &         & $2$   & $2$    & $1$ & $\T$       & 
		$-0.\overline{3}$ & 
		$0.\overline{7}$ & \num{6.0e0}    \\ \hline
		$100\T$   & $26$     & $1$  & 
		$\textcolor{regime}{01\T}\textcolor{fraction}{0}$            & 
		$01\T$  & $2$  
		& $0$ &         & $2$   & $2$    & $1$ & $0$        & 
		$0.0$             & 
		$1.0$            & \num{9.0e0}    \\ \hline
		$1000$    & $27$     & $1$  & 
		$\textcolor{regime}{01\T}\textcolor{fraction}{1}$            & 
		$01\T$  & $2$  
		& $0$ &         & $2$   & $2$    & $1$ & $1$        & 
		$0.\overline{3}$  & 
		$1.\overline{3}$ & \num{1.2e1}    \\ \hline
		$1001$    & $28$     & $1$  & 
		$\textcolor{regime}{010}\textcolor{exponent}{\T}$            & 
		$010$   & $3$  
		& $1$ & $\T$    & $4$   & $3$    & $0$ &            & 
		$0.0$             & 
		$1.0$            & \num{2.7e1}    \\ \hline
		$101\T$   & $29$     & $1$  & 
		$\textcolor{regime}{010}\textcolor{exponent}{0}$             & 
		$010$   & $3$  
		& $1$ & $0$     & $4$   & $4$    & $0$ &            & 
		$0.0$             & 
		$1.0$            & \num{8.1e1}    \\ \hline
		$1010$    & $30$     & $1$  & 
		$\textcolor{regime}{010}\textcolor{exponent}{1}$             & 
		$010$   & $3$  
		& $1$ & $1$     & $4$   & $5$    & $0$ &            & 
		$0.0$             & 
		$1.0$            & \num{2.4e2}    \\ \hline
		$1011$    & $31$     & $1$  & 
		$\textcolor{regime}{011}\textcolor{exponent}{\T}$            & 
		$011$   & $4$  
		& $2$ & $\T0$   & $10$  & $7$    & $0$ &            & 
		$0.0$             & 
		$1.0$            & \num{2.2e3}    \\ \hline
		$11\T\T$  & $32$     & $1$  & 
		$\textcolor{regime}{011}\textcolor{exponent}{0}$             & 
		$011$   & $4$  
		& $2$ & $00$    & $10$  & $10$   & $0$ &            & 
		$0.0$             & 
		$1.0$            & \num{5.9e4}    \\ \hline
		$11\T0$   & $33$     & $1$  & 
		$\textcolor{regime}{011}\textcolor{exponent}{1}$             & 
		$011$   & $4$  
		& $2$ & $10$    & $10$  & $13$   & $0$ &            & 
		$0.0$             & 
		$1.0$            & \num{1.6e6}    \\ \hline
		$11\T1$   & $34$     & $1$  & 
		$\textcolor{regime}{1\T\T}\textcolor{exponent}{\T}$          & 
		$1\T\T$ & $5$  
		& $3$ & $\T00$  & $28$  & $19$   & $0$ &            & 
		$0.0$             & 
		$1.0$            & \num{1.2e9}    \\ \hline
		$110\T$   & $35$     & $1$  & 
		$\textcolor{regime}{1\T\T}\textcolor{exponent}{0}$           & 
		$1\T\T$ & $5$  
		& $3$ & $000$   & $28$  & $28$   & $0$ &            & 
		$0.0$             & 
		$1.0$            & \num{2.3e13}   \\ \hline
		$1100$    & $36$     & $1$  & 
		$\textcolor{regime}{1\T\T}\textcolor{exponent}{1}$           & 
		$1\T\T$ & $5$  
		& $3$ & $100$   & $28$  & $37$   & $0$ &            & 
		$0.0$             & 
		$1.0$            & \num{4.5e17}   \\ \hline
		$1101$    & $37$     & $1$  & 
		$\textcolor{regime}{1\T0}\textcolor{exponent}{\T}$           & 
		$1\T0$  & $6$  
		& $4$ & $\T000$ & $82$  & $55$   & $0$ &            & 
		$0.0$             & 
		$1.0$            & \num{1.7e26}   \\ \hline
		$111\T$   & $38$     & $1$  & 
		$\textcolor{regime}{1\T0}\textcolor{exponent}{0}$            & 
		$1\T0$  & $6$  
		& $4$ & $0000$  & $82$  & $82$   & $0$ &            & 
		$0.0$             & 
		$1.0$            & \num{1.3e39}   \\ \hline
		$1110$    & $39$     & $1$  & 
		$\textcolor{regime}{1\T0}\textcolor{exponent}{1}$            & 
		$1\T0$  & $6$  
		& $4$ & $1000$  & $82$  & $109$  & $0$ &            & 
		$0.0$             & 
		$1.0$            & \num{1.0e52}   \\ \hline
		$1111$    & $40$     &      
		&                                                              
		&         
		&      &     &         &       &        &     &            
		&                   &                  & $\infty$       \\ 
		\hline
		\end{tabular}
		\egroup
	\end{center}
\end{table}
As an illustrative example, Table~\ref{tab:example} 
presents the decoding of all positive 
4-trit tekums, including a complete account of the intermediate quantities. 
Even at this limited size, tekums already exhibit a 
substantial dynamic range, approaching 
the recommended 
general-purpose range of $10^{\pm 55}$ derived in \cite{2024-takum}. 
Nevertheless, the 
format warrants further evaluation with respect to its numerical properties. 
This forms the focus of the following section.
\section{Evaluation}
Before evaluating tekums against other formats, it is first necessary to 
consider the general principles required for fair comparisons between binary 
and ternary systems.
\subsection{Comparing Binary with Ternary}
\label{sec:evaluation}
As noted in the introduction, a trit contains approximately 
$1.58$ bits of 
information. This difference must be accounted for when comparing binary and 
ternary formats. For example, when displaying a quantity relative to a bit 
count $n$, the ternary dataset must be scaled accordingly.
\par
A particular challenge arises when conducting benchmarks between formats, where 
such simple rescaling is not possible. The standard binary widths of number 
formats are $8$, $16$, $32$, and $64$ bits. A direct comparison between an 
$8$-bit format and an $8$-trit format would be misleading, as the latter 
possesses $3^8 = 6561$ representations in contrast to only $2^8 = 256$. By 
dividing the binary bit widths $8$, $16$, $32$, and $64$ by $\log_2(3) \approx 
1.58$, one obtains approximately $5.0$, $10.1$, $20.2$, and 
$40.4$, which are 
naturally matched by the integers $5$, $10$, $20$, and $40$. These represent 
the trit counts of ternary strings with approximately equivalent information 
content.
\par
The difficulty arises at $8$ bits: tekums are defined only for even trit 
counts, meaning a 5-trit tekum cannot be used. One possible approach is to 
evaluate both 4-trit and 6-trit variants and present both results, noting that 
an \enquote{imaginary} 5-trit tekum would fall somewhere in between. For higher 
precisions, however, suitable matches can be found at the even trit counts of 
$10$, $20$, and $40$, thereby enabling more direct 
comparisons.
\subsection{Format}
Posits consist of five bit fields (sign bit, regime 
bits, regime terminator 
bit, exponent bits, and fraction bits). Takums are 
structurally similar, 
also requiring five fields (sign bit, direction 
bit, regime bits, 
characteristic bits, and fraction bits). In 
contrast, tekums are considerably 
simpler, comprising only three trit fields: regime 
trits, exponent trits, and 
fraction trits. This simplicity may not only 
facilitate formal analysis but 
could also prove advantageous for hardware 
implementation.
\subsection{Accuracy}
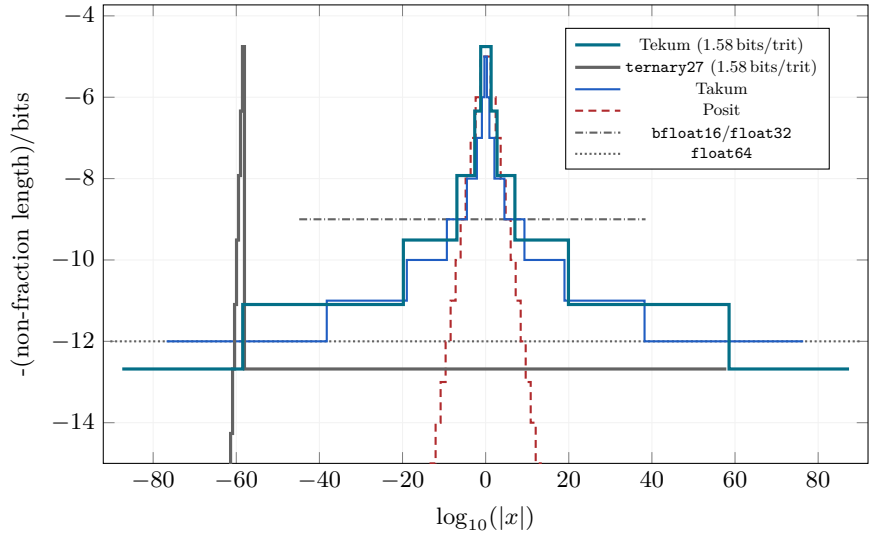
\begin{figure}[tbp]
	\begin{center}
		\begin{tikzpicture}
			\begin{axis}[
				reverse legend,
				scale only axis,
				width=\textwidth/1.2,
				height=\textwidth/2,
				xlabel={$\log_{10}(|x|)$},
				ylabel={-(non-fraction length)/bits},
				ymin=-15,
				xmin=-92,
				xmax=92,
				grid=major,
				grid style={line width=.1pt, draw=gray!10},
				legend style={nodes={scale=0.7, transform 
				shape}},
				legend style={at={(0.95,0.95)},anchor=north 
				east},
			]
				\addplot[const 
				plot,fraction,thick,densely
				dotted] table [x=exp10, y=float64, col 
				sep=comma]{code/exponent_cost-2-posit-float.csv};
				\addlegendentry{\texttt{float64}};

				\addplot[const plot,fraction,densely 
				dashdotted,thick] 
				table [x=exp10, 
				y=float32, 
				col 
				sep=comma]{code/exponent_cost-2-posit-float.csv};
				\addlegendentry{\texttt{bfloat16}/\texttt{float32}};

				\addplot[const 
				plot,sign,thick,densely
				 dashed] 
				table [x=exp10, 
				y=posit, col 
				sep=comma]{code/exponent_cost-2-posit-float.csv};
				\addlegendentry{Posit};

%
%

				\addplot[const 
				plot,t-exponent,thick]
				 table 
				[x=exp10, 
				y=lintakum, col 
				sep=comma]{code/exponent_cost-2-lintakum.csv};
				\addlegendentry{Takum};

				\addplot[const 
				plot,fraction,very 
				thick] 
				coordinates 
				{
					(-69.180, -42.7940) 
					(-68.860, -39.6241) 
					(-68.529, -38.0391) 
					(-68.052, -36.4541) 
					(-67.575, -34.8692) 
					(-67.098, -33.2842) 
					(-66.621, -31.6993) 
					(-66.144, -30.1143) 
					(-65.667, -28.5293) 
					(-65.190, -26.9444) 
					(-64.712, -25.3594) 
					(-64.235, -23.7744) 
					(-63.758, -22.1895) 
					(-63.281, -20.6045) 
					(-62.804, -19.0196) 
					(-62.327, -17.4346) 
					(-61.850, -15.8496) 
					(-61.373, -14.2647) 
					(-60.895, -12.6797) 
					(-60.418, -11.0947) 
					(-59.941,  -9.5098) 
					(-59.464,  -7.9248) 
					(-58.987,  -6.3399) 
					(-58.510,  -4.7549) 
					(-58.032, -12.6797) 
					(57.9078, -12.6797) 
				};
				\addlegendentry{\texttt{ternary27} 
				($1.58\,\text{bits}/\text{trit}$)};


				\addplot[const plot,regime,very thick] table 
				[x=exp10, y=tekum, col 
				sep=comma]{code/exponent_cost-3.csv};
				\addlegendentry{Tekum 
				($1.58\,\text{bits}/\text{trit}$)};
			\end{axis}
		\end{tikzpicture}
	\end{center}
	\caption{
		The number of non-fraction bits, which can be considered as 
		overhead,
		relative to the represented value 
		$x$ in a selection of floating-point
		formats. The y-axis is inverted, thus meaning that higher values
		mean less overhead.
	}
	\label{fig:accuracy}
\end{figure}
The first aspect of analysis concerns the accuracy 
of the tekum format. 
Following the approach taken in \cite{2024-takum}, one may consider the number 
of bits required to encode the sign and exponent of a given number $x > 0$, 
which can be interpreted as the overhead of encoding. The larger this overhead, 
the fewer fraction bits remain available. To enable comparison between binary 
and ternary formats, the overhead in trits for ternary formats is multiplied by 
$\log_2(3)$.
\par
The results are shown in Figure~\ref{fig:accuracy}. 
It is immediately evident 
that \texttt{ternary27} performs poorly, owing to its excessive waste of 
representations. The spike on the left is the result of a mechanism akin to 
subnormals; however, the effect is limited and occurs in an irrelevant region 
of numbers close to $10^{-60}$. 
\par
In contrast, tekums perform favourably. Compared to 
both posits and (linear) 
takums, tekums exhibit a region of highest accuracy 
of similar size to posits. 
This addresses a weakness of takums, which display 
a sharp drop in precision 
around the centre. Notably, tekums also achieve a 
significantly broader region 
of equal or superior precision compared with 
\texttt{bfloat16} and 
\texttt{float32}. In addition, tekums retain the same desirable logarithmic 
tapering property as takums.
\subsection{Dynamic Range}
\begin{figure}[tbp]
	\begin{center}
		\begin{tikzpicture}
			\begin{axis}[
				reverse legend,
				scale only axis,
				width=\textwidth/1.2,
				height=\textwidth/2,
				ymin=10^-95,
				ymax=10^95,
				xlabel={bit string length $n$},
				ylabel={$\log_{10}(\text{dynamic
				 range})$},
				ymode=log,
				ylabel shift=-0.1cm,
				xtick={2,4,8,16,32,64},
				xminorticks=true,
				yminorticks=true,
				ytick={10^-100,10^-80,10^-60,10^-40,10^-20,10^0,
					10^20,10^40,10^60,10^80,10^100},
				yticklabels={-100,-80,-60,-40,-20,0,20,40,60,80,100},
				grid=both,
				minor y tick num=0,
				grid style={line width=.1pt, draw=gray!10},
				major grid style={line width=.2pt,draw=gray!30},
				legend style={nodes={scale=0.7, transform 
				shape}},
				legend style={at={(0.97,0.5)},anchor=east}
			]
				\addplot[direction,mark=star,forget
				plot,thick,only 
				marks] 
				coordinates {(16,1.175494351e-38)};
				\addplot[direction,mark=star,thick,only
				 marks] 
				coordinates 
				{(16,3.38953139e38)};
				\addlegendentry{\texttt{bfloat16}};

				\addplot[fraction,mark=Mercedes
				 star,forget 
				plot,thick,only 
				marks] 
				coordinates {(8,6.103515625e-5)};
				\addplot[fraction,mark=Mercedes
				 star,thick,only 
				marks] 
				coordinates 
				{(8,57344.0)};
				\addlegendentry{\texttt{OFP8 E5M2}};

				\addplot[fraction,mark=x,forget
				 plot,thick,only 
				marks] 
				coordinates 
				{(8,0.016)};
				\addplot[fraction,mark=x,thick,only
				 marks] 
				coordinates 
				{(8,240)};
				\addlegendentry{\texttt{OFP8 E4M3}};

				\addplot[fraction,thick,densely
				dotted,mark=*,mark 
				options={solid,scale=0.3}]
				 table 
				[x=n, 
				y=ieee-subnormal-min, col sep=comma] 
				{code/dynamic_range.csv};
				\addlegendentry{IEEE 754 subnormal};

				\addplot[fraction,densely
				dashdotted,thick,mark=*,mark
				options={solid,scale=0.3}]
				table [x=n, 
				y=ieee-normal-min, col sep=comma] 
				{code/dynamic_range.csv};
				\addplot[fraction,forget
				 plot,densely 
				dashdotted,thick,mark=*,mark
				options={solid,scale=0.3}]
				table [x=n, 
				y=ieee-max, col 
				sep=comma] 
				{code/dynamic_range.csv};
				\addlegendentry{IEEE
				 754 normal};

				\addplot[sign,thick,mark=*,mark
				options={solid,scale=0.3}]
				 table [x=n, 
				y=posit2-min, col 
				sep=comma] {code/dynamic_range.csv};
				\addplot[sign,forget
				plot,thick,mark=*,mark
				options={solid,scale=0.3}]
				table [x=n, 
				y=posit2-max, col sep=comma] 
				{code/dynamic_range.csv};
				\addlegendentry{Posit};

				\addplot[t-exponent,thick,mark=*,mark
				options={solid,scale=0.3}]
				 table [x=n, 
				y=lintakum-min, 
				col sep=comma] 
				{code/dynamic_range.csv};
				\addplot[t-exponent,forget
				plot,thick,mark=*,mark
				options={solid,scale=0.3}]
				table [x=n, y=lintakum-max, col 
				sep=comma] {code/dynamic_range.csv};
				\addlegendentry{Takum};

%
%
%

				\addplot[fraction,mark=o,forget
				 plot,only 
				marks,very thick] 
				coordinates {(42.794,6.6e-70)};
				\addplot[fraction,mark=o,only
				 marks,very thick] 
				coordinates 
				{(42.794,8.1e57)};
				\addlegendentry{\texttt{ternary27} 
				($1.58\,\text{bits}/\text{trit}$)};
%

				\addplot[regime,very
				 thick,mark=*,mark 
				options={solid,scale=0.3}]
				 table [x=n, 
				y=tekum-min, col 
				sep=comma] {code/dynamic_range-3.csv};
				\addplot[regime,very
				 thick,forget 
				plot,mark=*,mark 
				options={solid,scale=0.3}]
				table 
				[x=n, y=tekum-max, 
				col sep=comma] 
				{code/dynamic_range-3.csv};
				\addlegendentry{Tekum
				($1.58\,\text{bits}/\text{trit}$)};
			\end{axis}
		\end{tikzpicture}
	\end{center}
	\caption{Dynamic range relative to the bit string length $n$ for tekum, 
	(linear) takum, posit and a selection of 
	floating-point formats.}
	\label{fig:dynamic_range}
\end{figure}
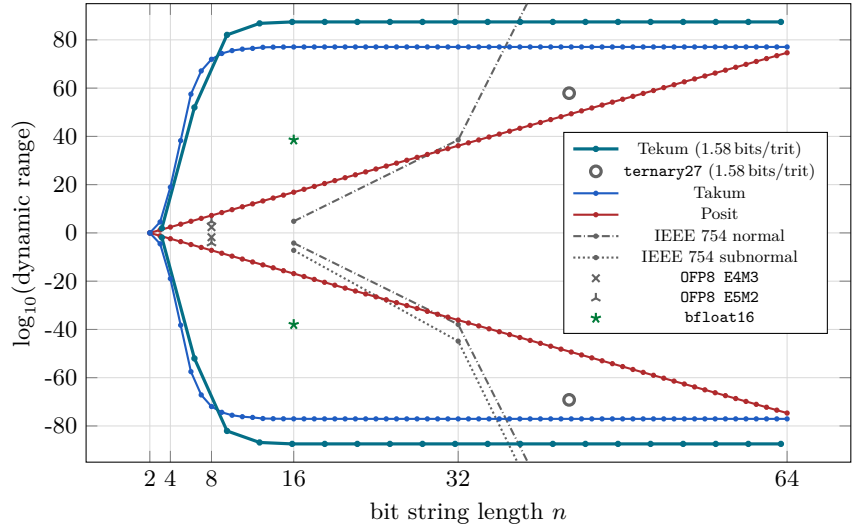%
The dynamic range results, shown in Figure~\ref{fig:dynamic_range}, largely 
confirm expectations, as the tekum dynamic range was explicitly designed to be 
$10^{\pm 87}$. Although this range exceeds the previously identified minimum 
recommended general-purpose dynamic range of $10^{\pm 55}$, it remains 
comfortably within the bound suggested by \textsc{Quevedo}. Owing to the 
tapered precision property, only a small fraction of values extend beyond 
$10^{\pm 55}$. 
\par
It is also noteworthy that, as with takums, the dynamic range of tekums 
increases rapidly at low precisions, quickly reaching its maximum, which is a 
desirable feature. The \texttt{ternary27} format 
also achieves a 
general-purpose dynamic range, but this range is asymmetric due to its 
inclusion of subnormal representations.
\subsection{Formal Analysis}
Building upon the discussion of the format's properties,
we now turn to the formal analysis and validation of the
tekum format in terms of its fundamental
characteristics, following the same methodology as applied
to the \emph{takum} format in \cite[Section~4.5]{2024-takum}.
We begin by showing that there exist no redundant encodings,
as stated in the following proposition.
\begin{proposition}[Uniqueness]
	Let $n \in 2\mathbb{N}_1$ and
	$\bm{t},\bm{u} \in \Tset_n$. It holds
	\begin{equation}
		\tekum_n(\bm{t}) = \tekum_n(\bm{u})
		\Rightarrow
		\bm{t} = \bm{u}.
	\end{equation}
\end{proposition}
\begin{proof}
	Assume $\tekum_n(\bm{t}) = \tekum_n(\bm{u})$. 
	By construction, it follows that $\bm{t} = \bm{u}$
	for the special cases 
	$\tekum_n(\bm{t}),\tekum_n(\bm{u}) \in 
	\{0,\infty,\mathrm{NaR}\}$. Hence, it remains to
	consider the non-special case, for which we define
	$\tekum_n(\bm{t}) := s \cdot (1+f) \cdot 3^e$ and
	$\tekum_n(\bm{u}) := \tilde{s} \cdot 
	(1+\tilde{f}) \cdot 3^{\tilde{e}}$.
	\par
	Given $\tekum_n(\bm{t}) = \tekum_n(\bm{u})$, we
	can immediately deduce $s = \tilde{s}$, since
	$(1+f)$, $(1+\tilde{f})$, $3^{\tilde{e}}$, 
	and $3^e$ are all positive. Moreover, as
	$(1+f),(1+\tilde{f}) \in (0.5,1.5)$, both
	$f=\tilde{f}$ and $e=\tilde{e}$ are uniquely
	determined. By construction, 
	$(s,e,f)=(\tilde{s},\tilde{e},\tilde{f})$ 
	implies 
	$\bm{t}=\bm{u}$.\qed
\end{proof}
Although this is a straightforward result, it highlights that, by
design, no ternary representations are wasted on redundant encodings.
We now proceed to show that ternary negation of a tekum
corresponds to negating its numerical value.
\begin{proposition}[Negation]\label{prop:negation}
	Let $n \in 2\mathbb{N}_1$ and
	$\bm{t} \in \Tset_n$ with $\tekum_n(\bm{t}) 
	\notin \{\mathrm{NaR},\infty\}$. It holds
	\begin{equation}
		\tekum_n(-\bm{t}) = 
		-\tekum_n(\bm{t}).
	\end{equation}
\end{proposition}
\begin{proof}
	The only special case, namely $\bm{t} 
	=0\cdots 0$, is trivial, given
	$\tekum_n(-0\cdots 0) = 
	\tekum_n(0\cdots 0) = 0$
	holds. Furthermore, it holds that
	\begin{equation}
		\anchor_n(-\bm{t}) = 
		|-\bm{t}| - 
		1\T\cdots1\T = |\bm{t}| - 
		1\T\cdots1\T =
		\anchor_n(\bm{t}),
	\end{equation}
	which implies that $\bm{t}$ and
	$-\bm{t}$ share the same
	\textcolor{regime}{$\bm{r}$},
	\textcolor{exponent}{$\bm{e}$},
	and \textcolor{fraction}{$\bm{f}$},
	and thus the same fraction value $f$ and
	exponent value $e$. The only distinction
	between them lies in their signs $s$, which
	are opposite, yielding the desired result.\qed
\end{proof}
Another result concerns the ternary monotonicity
of tekums, which is formalised as follows.
\begin{proposition}[Monotonicity]
	Let $n \in 2\mathbb{N}_1$ and
	$\bm{t},\bm{u} \in \Tset_n$. It holds
	\begin{equation}
		\integer_n(\bm{t}) < 
		\integer_n(\bm{u})
		\Rightarrow
		\tekum_n(\bm{t}) < \tekum_n(\bm{u}).
	\end{equation}
\end{proposition}
\begin{proof}
	We first consider the special cases. By convention, 
	$\mathrm{NaR}$ is smaller than any other represented 
	tekum value, which means that
	$\mathrm{NaR} < \tekum_n(\bm{u})$ holds. It also holds 
	$\integer_n(\T\cdots\T) < \integer_n(\bm{u})$, and 
	$\T\cdots\T$ is the ternary representation of 
	$\mathrm{NaR}$. Hence, monotonicity is satisfied for 
	$\mathrm{NaR}$. An analogous argument applies to 
	$\infty$, which by convention is larger than any other 
	represented tekum value and has the ternary 
	representation $1\cdots 1$, thereby establishing 
	monotonicity for $\infty$.
	\par
	We now turn to the non-special cases. By 
	Proposition~\ref{prop:negation}, it suffices to prove 
	the result for $\bm{t}$ and $\bm{u}$ with 
	$\integer_n(\bm{t}) > 0$ and $\integer_n(\bm{u}) > 0$, 
	as $\integer_n(0\cdots0) = 0$ and 
	$\tekum_n(0\cdots0) = 0 < 
	\tekum_n(\bm{t}), \tekum_n(\bm{u})$. 
	The objective can be further reduced to showing that, 
	for any $\bm{t} \in \{0\cdots01,\dots,1\cdots01\}$,
	\begin{equation}
		\tekum_n(\bm{t}) < 
		\tekum_n(\bm{t} + 0\cdots01),
	\end{equation}
	i.e.\ the successive monotonicity of the positive 
	non-special cases.
	\par

	\par
	We denote
	$\anchor_n(\bm{t}) := 
	\textcolor{regime}{\bm{r}} \mdoubleplus
	\textcolor{exponent}{\bm{e}}
	 \mdoubleplus 
	\textcolor{fraction}{\bm{f}}$ and
	$\anchor_n(\bm{t} + 0\cdots01) := 
	\textcolor{regime}{\bm{\tilde{r}}} 
	\mdoubleplus
	\textcolor{exponent}{\bm{\tilde{e}}}
	 \mdoubleplus 
	\textcolor{fraction}{\bm{\tilde{f}}}$,
	where tildes denote the components of the incremented 
	ternary vector. For $\integer_n(\bm{t}) > 0$, we have 
	$|\bm{t}| = \bm{t}$, and in particular, provided there 
	is no overflow,
	\begin{equation}
		\anchor_n(\bm{t} + 0\cdots01) = 
		\anchor_n(\bm{t}) + 0\cdots01.
	\end{equation}
	Hence, the field 
	$\textcolor{regime}{\bm{r}}
	\mdoubleplus
	\textcolor{exponent}{\bm{e}}
	\mdoubleplus
	\textcolor{fraction}{\bm{f}}$
	is incremented, and to establish monotonicity, we must
	examine how this affects 
	$\textcolor{regime}{\bm{r}}$,
	$\textcolor{exponent}{\bm{e}}$ and
	$\textcolor{fraction}{\bm{f}}$.
	\par
	Although we work in balanced ternary, carries may still
	occur, analogous to the binary case. We therefore
	distinguish cases depending on whether the increment
	affects only 
	$\textcolor{fraction}{\bm{f}}$, 
	or whether a carry propagates into 
	$\textcolor{exponent}{\bm{e}}$.
	The latter occurs precisely when
	$\textcolor{fraction}{\bm{f}} = 1\cdots 1$,
	in which case we have
	$\textcolor{fraction}{\bm{\tilde{f}}} = 
	\T\cdots \T$, leading to an additional case 
	distinction on 
	$\textcolor{exponent}{\bm{e}}$.
	In all cases, $\tilde{s} = s$.
	\par
	Following this scheme, as the first case we 
	assume $\textcolor{fraction}{\bm{f}} 
	\neq 1\cdots1$. It then holds
	$\textcolor{regime}{\bm{\tilde{r}}} =
	\textcolor{regime}{\bm{r}}$,
	$\textcolor{exponent}{\bm{\tilde{e}}} = 
	\textcolor{exponent}{\bm{e}}$ and
	$\tilde{f} = f + 3^{-p} > f$, in particular 
	$\tilde{e}=e$. It follows that
	\begin{equation}
	\tekum_n(\bm{t} + 0\cdots01) =
		\tilde{s} \cdot (1+\tilde{f}) 
		3^{\tilde{e}} =
		s \cdot (1 + \tilde{f}) 3^e >
		s \cdot (1 + f) 3^e = 
		\tekum_n(\bm{t}),
	\end{equation}
	which was to be shown.
	\par
	The second case is 
	$\textcolor{fraction}{\bm{f}} =
	1\cdots1$ and $\textcolor{exponent}{\bm{e}} 
	\neq 1\cdots 1$. Then
	$\textcolor{fraction}{\bm{\tilde{f}}} =
	\T\cdots\T$ and 
	$\textcolor{regime}{\bm{\tilde{r}}} = 
	\textcolor{regime}{\bm{r}}$,
	which implies $\tilde{r} = r$, $\tilde{c} = c$,
	$\tilde{p}=p$, $\tilde{b}=b$, and 
	$\tilde{e} = e + 1$ since the exponent field
	increments without overflowing. As
	$\textcolor{fraction}{\bm{\tilde{f}}} = 
	-\textcolor{fraction}{\bm{f}}$, we have
	$\tilde{f} = -f$, and thus
	\begin{equation}
		\tekum_n(\bm{t} + 0\cdots01) =
			\tilde{s} \cdot (1+\tilde{f}) 
			3^{\tilde{e}} =
			s \cdot (1 - f) 3^{e+1} >
			s \cdot (1 + f) 3^e = 
			\tekum_n(\bm{t}),
	\end{equation}
	as required.
	\par
	The third case is
	$\textcolor{regime}{\bm{r}} \neq 1\cdots 1$,
	$\textcolor{exponent}{\bm{e}} = 1\cdots 1$
	and
	$\textcolor{fraction}{\bm{f}} = 1\cdots 1$. 
	Then
	$\textcolor{fraction}{\bm{\tilde{f}}} =
	\T\cdots\T$ (so $\tilde{f} = -f$) and 
	$\textcolor{exponent}{\bm{\tilde{e}}} =
	\T\cdots\T$. In particular, $\tilde{r} = r+1$,
	hence $\tilde{b} > b$. By construction, the bias
	is chosen such that each exponent range is
	distinct for every regime. We can therefore 
	deduce $\tilde{e} > e$, and it follows that
	\begin{equation}
		\tekum_n(\bm{t} + 0\cdots01) =
		s \cdot (1 - f) 3^{\tilde{e}} \ge
		s \cdot (1 - f) 3^{e+1} >
		s \cdot (1 + f) 3^e = 
		\tekum_n(\bm{t}),
	\end{equation}
	as required.
	\par
	There is no fourth case where all trit fields
	are $1\cdots1$, as $\bm{t}$ is explicitly at
	most $1\cdots01$. Thus, all possibilities have
	been exhausted, completing the proof.
	\qed
\end{proof}
The significance of this result lies in the fact that there
exists a direct correspondence between the ordering of ternary
integer and tekum representations. This property, also present
for posits and takums, offers a notable implementation advantage:
the same logic used for integer comparison can be employed for
floating-point comparison, even eliminating the need for explicit
special-case handling of $\mathrm{NaR}$ and $\infty$.
\par
A particularly noteworthy aspect in the context of tekums
is that both $\mathrm{NaR}$ and $\infty$ integrate seamlessly into
this framework. Since tekums constitute the first balanced ternary
real number system, the introduction of a second special value,
$\infty$, in addition to $\mathrm{NaR}$, posed a potential risk of
breaking monotonicity. Remarkably, however, the structure remains
fully consistent, an elegant and non-trivial property of the format.
\par
We now turn to the next aspect, namely rounding.
Balanced ternary exhibits the elegant property
that rounding an integer is equivalent to simply truncating it.
This can be seen by considering an arbitrary real number expressed in
balanced ternary as the sum of an integral and a fractional part,
\begin{equation}
	\sum_{i=0}^{\infty} \bm{a}_i 3^i +
	\sum_{i=0}^{\infty} \bm{b}_i 3^{-(i+1)},
\end{equation}
with $\bm{a}, \bm{b} \in \Tset_\infty$, which we wish
to round to the nearest integer. The outcome depends solely on the
fractional part, which in the binary case lies in the interval
$[0,1)$, necessitating a distinction between the intervals $[0, 0.5)$,
exactly $0.5$, and $(0.5, 1)$. 
\par
In the balanced ternary case, however, the fractional part lies in the
interval $(-0.5, 0.5)$, meaning the nearest integer is always the integral
part $\sum_{i=0}^{\infty} \bm{a}_i 3^i$. In other words, rounding
amounts to unconditionally truncating the fractional part.
The question is how this property translates to tekums, which would
benefit from such stable rounding behaviour. For this, we provide the
following:
\begin{proposition}[Truncation is Rounding]
	Let $n \in 2\mathbb{N}_3$, $\bm{t} \in
	\Tset_n$ with $\tekum(\bm{t}) \notin 
	\{\mathrm{NaR},\infty\}$ and $\bm{a} := 
	\anchor_n(\bm{t})
	\in \Tset_n$. It holds
	\begin{equation}
		\anchor_n^\text{inv}(\bm{a}_{n-1} 
		\cdots \bm{a}_2) =
		\argmin_{\bm{u} \in 
		\Tset_{n-2}}(|\tekum(\bm{t}) -
		\tekum(\bm{u})|).
	\end{equation}
\end{proposition}
In other words, for a given $n$-trit tekum, the closest $(n-2)$-trit
tekum is obtained by truncating its anchor by two trits and converting
it back to a tekum (using $\bm{t}$'s sign). Since the conversion
between a tekum and its anchor is lossless, this property generalises
to any truncation of the anchor, provided the anchor remains at least
four trits long to avoid truncating the regime trits.
\begin{proof}
	Within the scope of the anchor, we are not restricted to even trit 
	counts, and may therefore consider the effect of truncating a single 
	trit. Let
	\begin{equation}
		\bm{a} = \textcolor{regime}{\bm{r}} \mdoubleplus 
		\textcolor{exponent}{\bm{e}} \mdoubleplus 
		\textcolor{fraction}{\bm{f}},
	\end{equation}
	and denote all truncated quantities with a tilde. Since truncation 
	never occurs within the regime trits, we only distinguish between 
	truncating an exponent trit and truncating a fraction trit. In both 
	cases, the sign $s$ remains unchanged.
	\par
	In the first case, truncating a trit from the fraction yields 
	$\textcolor{fraction}{\bm{\tilde{f}}}$, which represents the closest 
	integer to $\textcolor{fraction}{\bm{f}}$. Consequently, $\tilde{f}$ is 
	the closest representable value to $f$, and since $e$ is unchanged, the 
	tekum value
	\begin{equation}
		s \cdot (1+\tilde{f}) \cdot 3^{\tilde{e}} = s \cdot 
		(1+\tilde{f}) \cdot 3^e
	\end{equation}
	is optimally close to $\tekum(\bm{t})$, as required.
	\par
	In the second case, we truncate a trit from the exponent trits.
	It holds $f = \tilde{f} = 0$. As the regime trits remain unchanged, 
	both the bias and the number of exponent trits are preserved
	($b = \tilde{b}$, $c = \tilde{c}$). The truncated exponent
	\begin{equation}
		\tilde{e} = \integer_c(\textcolor{exponent}{\bm{\tilde{e}}}) + b
	\end{equation}
	is the closest representable integer to $e$, and the tekum value $s 
	\cdot 3^{\tilde{e}}$ is therefore the optimal approximation of 
	$\tekum(\bm{t})$.
	\qed
\end{proof}
To give an example, consider rounding
$\bm{t} = 01\T\T\T1\T\T$ 
($\anchor_n(\bm{t})=001\T1110$, 
$s = 1$, $r = 1$, $c = 0$, $p=5$, $b = 1$,
$e = 1$, $f = 3^{-5} \cdot (-42) \approx -0.173,
\tekum(\bm{t})\approx 1.654$)
to four trits. We truncate the anchor to the four
trits $001\T$ and obtain the tekum
$1\T11$, which represents the value $2.0$ and is 
the closest representation (see 
Table~\ref{tab:example}).
\par
This property has far-reaching implications despite the small
computational overhead of determining the anchor. Conversion between
precisions is essentially cost-free, requiring only the truncation or
extension of the anchor. Well-known issues in binary floating-point
arithmetic, such as inaccuracies due to \enquote{double rounding},
do not arise here: Truncation can be performed in any number of steps
with identical results.
\subsection{Hardware Implementation}
Ternary logic hardware remains in its infancy and is largely orthogonal to 
contemporary chip design. The theoretical advantages of ternary logic are 
substantial, offering higher information density, owing to its superior radix 
economy, which is particularly relevant in the context of the memory wall that 
constrains computer performance, and consequently reduced circuit and 
interconnect complexity, lower power dissipation, and potentially faster 
operational speeds.
\par
Historically, apart from binary computers that merely emulate trits using two 
bits each, such as the Setun and Setun~70 computers \cite{2011-setun}, 
proposals for genuine ternary computers have relied on alternative 
technologies. These include \textsc{Josephson} junctions 
\cite{1998-hardware-josephson} and optical computing, where dark represents 
zero and two orthogonal polarisations of light encode $\T$ and $1$ 
\cite{2003-hardware-optical}. However, such approaches never progressed beyond 
experimental prototypes.
\par
Ternary computing has recently attracted renewed interest for two main reasons: 
advances in ternary large-language models (LLMs) and developments in Carbon 
Nanotube Field-Effect Transistors (CNTFETs). In the former case, single trits 
are used to represent weights within deep neural networks, with notable 
examples including BitNet \cite{2023-bitnet-1,2025-bitnet-2} and proposals for 
ternary hardware implementations of neural networks \cite{2025-hardware-llm}. 
Regarding CNTFETs, hundreds of publications in recent years have addressed 
advances ranging from general logic gates 
\cite{2024-hardware-1,2024-hardware-3} to arithmetic units such as adders 
\cite{2024-hardware-2-adder}. 
\par
In the context of implementing takums, a particular significance is attributed 
to an adder dedicated to adding $\T1\cdots\T1$ for computing the anchor 
function $\anchor_n$, which is central to decoding and rounding 
tekums. This capability could enable more efficient circuit designs tailored to 
the tekum format.
\par
It is noteworthy that the recent AI revolution may not only have catalysed the 
exploration of novel arithmetic formats, but could also indirectly foster a 
transition from binary to ternary computing.
\section{Conclusion}
This paper has introduced \emph{tekum}, a novel balanced ternary real 
arithmetic format, derived from three fundamental design challenges and 
evaluated against rigorous criteria for fair comparison between binary and 
ternary representations.
\par
Our analysis has demonstrated that tekums possess highly favourable numerical 
properties, not only in terms of precision and dynamic range, but also in 
offering features unique to the ternary domain. Unlike their binary 
counterparts, posits and takums, tekums simultaneously accommodate both 
$\infty$ and $\mathrm{NaR}$, while retaining the simplicity of negation by 
flipping the underlying trit string. Perhaps most strikingly, tekums enable 
rounding by truncation, a property that eradicates 
at a stroke some notorious 
problems of rounding in binary arithmetic: double 
rounding errors, cascading carries in hardware, and 
the attendant inefficiencies.
\par
Taken together, these properties position tekums not merely as a curiosity, but 
as a radical and compelling alternative to established binary real number 
formats. They demonstrate that balanced ternary arithmetic is not only viable, 
but potentially transformative. While substantial work remains, tekums stand as 
a bold step towards a new numerical foundation for future ternary computing, a 
foundation that could redefine the very architecture of computation itself.
\label{sec:conclusion}
\printbibliography
\end{document}